\documentclass[11pt]{amsart}
\usepackage[margin=1in, marginparwidth=10mm,marginparsep=5mm]{geometry}
\usepackage[english]{babel}
\usepackage[ansinew]{inputenc}
\usepackage{amsmath, amsthm, amsfonts, amssymb}
\usepackage{latexsym}
\usepackage{bbm}
\usepackage{mathrsfs}
\usepackage{tikz}
\usetikzlibrary{decorations.pathreplacing}
\usepackage{graphicx}
\usepackage{multirow}
\usepackage{subfig}
\usepackage{youngtab}
\usepackage{stmaryrd}
\usepackage{ytableau}
\usepackage{hyperref}

\def\N{\mathbb{N}}

\def\C{\mathbb{C}}

\def\det{\mathrm{det}}
\def\per{\mathrm{per}}
\def\sgn{\mathrm{sgn}}
\def\GL{\mathrm{GL}}
\def\SL{\mathrm{SL}}
\def\VNP{\mathrm{VNP}}

\def\VPs{\mathrm{VP}_{\mathrm{ws}}}
\def\dc{\mathrm{dc}}

\def\Sym{\mathsf{Sym}}
\newcommand{\Det}{\Omega} 
\newcommand{\ol}[1]{\overline{#1}}

\def\la{\lambda}

\def\ot{\otimes} 
\def\sP{\#\mathrm{P}}

\def\id{\mathrm{id}}
\def\diag{\mathrm{diag}}

\def\ocp{Z}
\newcommand{\kron}{k}
\newcommand{\gctm}{K}

\newcommand{\out}{d}
\newcommand{\inn}{n}
\newcommand{\imm}{m}

\newcommand{\IC}{\ensuremath{\mathbb{C}}}

\newcommand{\IN}{\ensuremath{\mathbb{N}}}
\newcommand{\IZ}{\ensuremath{\mathbb{Z}}}
\newcommand{\IQ}{\ensuremath{\mathbb{Q}}}

\newcommand{\aS}{\ensuremath{\mathfrak{S}}}

\newcommand{\sV}{\ensuremath{\mathscr{V}}}

\newcommand{\tensor}{\smash{\textstyle\bigotimes}}
\newcommand{\HWV}{\mathsf{HWV}}

\newcommand{\val}{\mathrm{val}}
\newcommand{\PS}{\mathrm{PS}}

\swapnumbers
\numberwithin{equation}{section}
\newtheorem{theorem}{Theorem}[section]
\newtheorem{corollary}[theorem]{Corollary}
\newtheorem{lemma}[theorem]{Lemma}

\newtheorem{proposition}[theorem]{Proposition}
\newtheorem{conjecture}[theorem]{Conjecture}
\newtheorem{claim}[theorem]{Claim}

\newtheorem{definition}[theorem]{Definition}

\title[No occurrence obstructions in geometric complexity theory]{No
  occurrence obstructions in geometric\\ complexity theory$^*$}
\thanks{$^*$A preliminary version of this paper was presented at the
  57th Annual Symposium on Foundations of Computer Science (FOCS 2016), see \cite{BIP:16}}
\author[Peter B\"urgisser, Christian Ikenmeyer, Greta Panova]{Peter B\"urgisser$^1$, Christian Ikenmeyer$^2$, Greta Panova$^3$}
\thanks{$^1$Technische Universit\"at Berlin,  pbuerg@math.tu-berlin.de. Partially supported by DFG grant BU 1371/3-2.}
\thanks{$^2$Max Planck Institute for Informatics, Saarland Informatics Campus, cikenmey@mpi-inf.mpg.de} %Texas A\&M University, ciken$@$math.tamu.edu}
\thanks{$^3$University of Pennsylvania and Institute for Advanced Study, panova$@$math.upenn.edu. Partially supported by NSF grant DMS-1500834}
\date{\today}



\begin{document}
\raggedbottom

\begin{abstract}
The permanent versus determinant conjecture is a major problem in 
complexity theory that is equivalent to the separation of the
complexity classes $\VPs$ and $\VNP$.  
Mulmuley and Sohoni~\cite{gct1} suggested to 
study a strengthened version of this conjecture over the 
complex numbers that amounts to separating the orbit closures of 
the determinant and padded permanent polynomials. 
In that paper it was also proposed to separate these orbit closures 
by exhibiting {\em occurrence obstructions}, which are 
irreducible representations of $\GL_{n^2}(\C)$, which 
occur in one coordinate ring of the orbit closure, but not 
in the other. 
We prove that this approach is impossible. 
However, we do not rule out the general approach to the 
permanent versus determinant problem via 
{\em multiplicity obstructions} as proposed in~\cite{gct1}. 
\end{abstract}

\maketitle

\section{Introduction} 

A central problem in algebraic complexity theory is to prove that 
there is no efficient algorithm to evaluate the {\em permanent} 
$$
 \per_n := \sum_{\pi\in \aS_n} X_{1\pi(1)} \cdots X_{n\pi(n)}.
$$
The natural model of computation to study this question 
is the one of straight-line programs (or arithmetic circuits),
which perform arithmetic operations $+,-,*$ 
in the polynomial ring, 
starting with the variables $X_{ij}$ and complex constants. 
Efficient means that the number of arithmetic operations 
is bounded by a polynomial in $n$. 
The permanent arises in combinatorics and physics as a 
generating function. Its relevance for complexity theory 
derives from Valiant's discovery~\cite{Val:79a,Val:79}
that the evaluation of the permanent is a complete problem 
for the complexity class $\VNP$ 
(and also for the class $\sP$ in the model of Turing machines); 
see \cite{buer:00-3,mapo:08} for more information.   

The {\em determinant}  
\begin{equation*}%\label{def:det}
 \det_n 
  := \sum_{\pi\in \aS_n} \sgn(\pi)\, X_{1\pi(1)} \cdots X_{n\pi(n)} 
\end{equation*}
is known to have an efficient algorithm.  
Its evaluation is complete for the complexity class $\VPs$; 
see~\cite{Val:79a,toda:92}. From the definition it is clear that 
$\VPs\subseteq \VNP$ and proving the separation 
$\VPs\ne\VNP$ is the flagship problem in algebraic complexity theory. 
It can be seen as an ``easier'' version of the famous 
$\mathsf{P}\ne\mathsf{NP}$ problem; see~\cite{buerg:00}.

The conjecture $\VPs\ne\VNP$ can be restated without any reference to complexity
classes by directly comparing permanents and determinants. 
The {\em determinantal complexity} $\dc(f)$ of a polynomial 
$f\in\C[X_1,\ldots,X_N]$ is defined as the smallest
nonnegative integer $n\in\N$ such
that $f$ can be written as a determinant of an $n$ by $n$ matrix whose 
entries are affine linear forms in the variables~$X_i$. 
Valiant~\cite{Val:79a} proved that the determinant is computationally universal 
in the sense that $\dc(f) \le n$ 
if $f$ can be written as an arithmetic expression involving less than $n$ operations $+,-,*$.
Moreover, 
Valiant~\cite{Val:79a,vali:82} and Toda~\cite{toda:92} 
proved that $\VPs\ne\VNP$ is equivalent to the  
following conjecture. 

\begin{conjecture}[Valiant 1979]\label{conj:dc}
The determinantal complexity $\dc(\per_n)$ grows
superpolynomially in~$n$.
\end{conjecture}

It is known~\cite{grenet:11} that $\dc(\per_n) \le 2^n -1$.  
Finding lower bounds on $\dc(\per_n)$ is an active area of research
\cite{mignon-ressayre:04,cai:2010,lamare:13,hu-ik:14,al-bo-ve:15,yabe:15},
but the best known lower bounds are only $\Omega(n^2)$.

\subsection{An attempt via algebraic geometry and representation theory}

Towards answering Conjecture~\ref{conj:dc}, 
Mulmuley and Sohoni~\cite{gct1,gct2,mulmuley:11} proposed an approach based on algebraic geometry and
representation theory, for which they coined the name geometric complexity theory. 

We denote by $\Sym^n V^*$ the space of homogeneous polynomial functions of degree~$n$ on 
a finite dimensional complex vector space $V$.
The group $G:=\GL(V)$ acts on $\Sym^n V^*$ in the canonical way: %by linear substitution:
$(g \cdot f)(v) := f(g^{-1}v)$ for $g\in G$, $f\in \Sym^n V^*$, $v\in V$.  
We denote by $G\cdot f :=\{ gf \mid g\in G\}$ the {\em orbit} of~$f$.
% is obtained by applying all possible invertible linear transformations to $f$.  
We assume now $V:=\C^{n\times n}$, view the determinant $\det_n$ as an element of $\Sym^n V^*$,
and consider its {\em orbit closure}: 
%an element of the space $\Sym^n V^*$ of homogeneous polynomials of degree~$n$ on $V$. 
%and view $\det_n\in \Sym^n V^*$. 
%The group $G:=\GL(V)$ acts on $\Sym^n V^*$ in the canonical way by linear substitution:
%$g f(v) := f(g^{-1}v)$ for $g\in G$, $f\in \Sym^n V^*$, $v\in V$.  
%The {\em orbit} $G\cdot \det_n$ is obtained by applying all possible 
%invertible linear transformations to $\det_n$.  
%onsider the closure 
\begin{equation}\label{def:Omegan}
\Det_n := \ol{\GL_{n^2}\cdot \det_n} \subseteq \Sym^n (\C^{n\times n})^* 
\end{equation}
with respect to the Euclidean topology. (By a general principle, 
this is the same as the closure with respect to the Zariski topology; 
see~\cite[\S2.C]{mumf:95}.) 
It is easy to see that $\Det_2 = \Sym^2(\C^{2\times 2})^*$. 
For $n=3$, the boundary of $\Det_n$ has been determined recently~\cite{lahu:16}, 
but for $n=4$ it is already unknown. 

For $n > m$ we consider the {\em padded permanent} defined as 
$X_{11}^{n-m} \per_m \in\Sym^n (\C^{m\times m})^*$, where $X_{11}$ denotes the linear form 
providing the $(1,1)$-entry of a matrix in $\C^{m\times m}$. 
Via the standard projection $\C^{n\times n} \to \C^{m\times m}$, we can view 
$X_{11}^{n-m} \per_m$ as an element of the bigger space $\in\Sym^n (\C^{n\times n})^*$.
(Sometimes the padding is achieved by using a linear form different from $X_{11}$, e.g., 
%a variable not appearing in $\per_m$, 
but this is irrelevant, see~\cite[Appendix]{ik-panova:15}.) 

The following conjecture was stated in \cite{gct1}.
We refer to \cite[Prop.~9.3.2]{BLMW:11} for an equivalent formulation 
in terms of complexity classes that goes back to \cite{buerg:04}. 
(In particular see \cite[Problem~4.3]{buerg:04}.)  

\begin{conjecture}[Mulmuley and Sohoni 2001]\label{conj:dc-bord}
For all $c\in\N_{\ge 1}$ 
we have $X_{11}^{m^c-m} \per_m \not\in \Det_{m^c}$ 
for infinitely many~$m$.
\end{conjecture}

%This conjecture was stated in \cite{gct1}.

Conjecture~\ref{conj:dc-bord} implies Conjecture~\ref{conj:dc}. 
Indeed, using that 
$\GL_{n^2}$ is dense in $\C^{n^2\times n^2}$, 
one shows (e.g., see~\cite{buerg-survey:15})  
that $\dc(\per_m) \le n$ implies $X_{11}^{n-m} \per_m \in \Det_n$. 
(The latter must be a point in the boundary of $\Det_n$ if $n>m$.) 

The following strategy towards Conjecture~\ref{conj:dc-bord} was proposed in \cite{gct1}. 
The action of the group $G=\GL(V)$ on 
$\Sym^n V^*$ induces an action on its (graded) coordinate ring 
$\C[\Sym^n V^*]=\oplus_{d\in\N} \Sym^d \Sym^n V$. 
The homogeneous parts $\Sym^d \Sym^n V$ are called {\em plethysms}:  
in fact, we obtain the natural $G$-action on theses spaces. 
%Since $\Det_n$ is $G$-invariant, 
%The $G$-action descends on the
The coordinate ring $\C[\Det_n]$ of the orbit closure $\Det_n$ 
is obtained as the homomorphic image
of $\C[\Sym^n V^*]$ via the restriction of regular functions, 
and the $G$-action descends on this. 
In particular, we obtain the degree~$d$ part $\C[\Det_n]_d$ of $\C[\Det_n]$
as the homomorphic $G$-equivariant image of $\Sym^d \Sym^n V$. 
%on $\Det_n$ induces a corresponding action 
%on the coordinate ring $\C[\Det_n]$. 

It is well known~\cite{FH:91} that 
the irreducible polynomial representations of $G$ can be labeled by 
partitions $\la$ into at most $n^2$ parts. 
The coordinate ring $\C[\Det_n]$ is a direct sum of its 
irreducible submodules since $G$ is reductive. 
We say that $\la$ occurs in $\C[\Det_n]$ if 
it contains %the dual of 
an irreducible $G$-module of type $\la$. 
%(It is useful to identify spaces with their duals to avoid negative weights.)

Let $\ocp_{n,m}$ denote the orbit closure of the padded permanent ($n>m$): 
\begin{equation}\label{def:Znm}
 \ocp_{n,m} := \ol{\GL_{n^2}\cdot X_{11}^{n-m} \per_m} \subseteq \Sym^n (\C^{n\times n})^*.
\end{equation}
If $X_{11}^{n-m} \per_m$ is contained in $\Det_n$, then 
$\ocp_{n,m} \subseteq \Det_n$, and the restriction defines 
a surjective $G$-equivariant homomorphism 
$\C[\Omega_n]\to\C[\ocp_{n,m}]$ of the coordinate rings. 
Schur's lemma implies that 
if $\la$ occurs in $\C[\ocp_{n,m}]$, 
then it must also occur in $\C[\Det_n]$. 
A partition $\la$ violating this condition is called an 
{\em occurrence obstruction}. Its existence thus 
proves that $\ocp_{n,m} \not\subseteq \Det_n$ 
and hence $\dc(\per_m)>n$.
It is known that occurrence obstructions~$\la$ 
must satisfy $|\la| = nd$ and $\ell(\la) \le m^2$, 
see~\cite{gct1,gct2,BLMW:11}.
Here $|\la| := \sum_i \la_i$ denotes the {\em size} of $\la$ 
and $\ell(\la)$ denotes the {\em length} of $\la$, 
which is defined as the number of 
nonzero parts of $\la$.
We write $\la \vdash |\la|$, so in our case $\la \vdash nd$.

In \cite{gct1,gct2} it was suggested to prove Conjecture~\ref{conj:dc-bord} 
by exhibiting occurrence obstructions. 
More specifically, the following conjecture was put forth. 

\begin{conjecture}[Mulmuley and Sohoni 2001]\label{conj:occ-obstr}
For all $c\in\N_{\ge 1}$, 
for infinitely many $m$, 
there exists a partition~$\la$ 
occurring in $\C[Z_{m^c,m}]$ but not in $\C[\Det_{m^c}]$.
\end{conjecture}

This conjecture implies Conjecture~\ref{conj:dc-bord}
by the above reasoning. 

Conjecture~\ref{conj:occ-obstr} on the existence of occurrence obstructions  
has stimulated a lot of research and 
has been the main focus of researchers 
in geometric complexity theory
in the past years, see Section~\ref{se:prev-work}.
Unfortunately, this conjecture is false! This is the main result of this work.
More specifically, we show the following. 

\begin{theorem}\label{th:main}
Let $n,d,m$ be positive integers with 
$n \ge m^{25}$ and $\la\vdash nd$. If $\la$ occurs in 
$\C[Z_{n,m}]$, then $\la$~also occurs in $\C[\Omega_n]$. 
In particular, Conjecture~\ref{conj:occ-obstr} is false. 
\end{theorem}

One can likely improve the bound on~$n$ by a more careful analysis.

%%%
\subsection{Related work: Kronecker coefficients}\label{se:prev-work}

Kronecker coefficients are fundamental quantities that 
have been the object of study in algebraic combinatorics 
for a long time~\cite{murna:38}. 
A difficulty in their study is that there is no known counting interpretation 
of them~\cite{stanley:00}. 
The {\em Kronecker coefficient} $\kron(\la,\mu,\nu)$ 
of three partitions $\la,\mu,\nu$ of the same size~$d$
is defined as the dimension of the space of $\aS_d$-invariants 
of $[\la]\ot [\mu]\ot [\nu]$, where $[\la]$ denotes the 
irreducible $\aS_d$-module of type $\la$
and $\aS_d$ is the symmetric group on $d$ symbols;  
see \cite[I\S7, internal product]{macdonald:95}. 
We write $n\times d$ for the rectangular 
partition $(d,\ldots,d)$ of size $nd$ and call 
$\kron_n(\la) := \kron(\la,n\times d,n\times d)$ 
the {\em rectangular Kronecker coefficient} 
of $\la\vdash nd$. 
In \cite{ik-m-w:15} it was shown that deciding positivity of 
Kronecker coefficients in general is NP-hard, but this proof 
fails for rectangular formats. 

Let $\gctm_n(\la)$ denote the multiplicity by which the irreducible $\GL_{n^2}$-module of type 
$\la\vdash nd$ occurs in $\C[\Omega_n]_d$. 
We call the numbers $\gctm_n(\la)$ the {\em GCT-coefficients}.   
%Note that an occurrence obstruction for 
%$Z_{n,m}\not\subseteq\Omega_n$ 
%is a partition $\la$ for which  $\gctm_n(\la)=0$ 
%and such that $\la$ occurs in $\C[Z_{n,m}]$. 
In \cite{gct1} it was realized that GCT-coefficients 
can be upper bounded by rectangular Kronecker coefficients:
we have $\gctm_n(\la) \le \kron_n(\la)$ for $\la\vdash nd$. 
In fact, the multiplicity of $\la$ in the coordinate ring of the orbit 
$\GL_{n^2}\cdot\det_n$ equals the so-called symmetric 
rectangular Kronecker coefficient~\cite{BLMW:11}, which is  
upper bounded by $\kron_n(\la)$. 

Note that an occurrence obstruction for 
$Z_{n,m}\not\subseteq\Omega_n$ 
is a partition $\la$ for which  $\gctm_n(\la)=0$ 
and such that $\la$ occurs in $\C[Z_{n,m}]$. 
Since hardly anything was known about the GCT-coefficients, 
it was proposed in \cite{gct1} to find $\la$ 
for which the Kronecker coefficient $k_n(\la)$ vanishes 
and such that $\la$ occurs in $\C[Z_{n,m}]$. 
This potential method of proving Conjecture~\ref{conj:occ-obstr},
has stimulated research in algebraic combinatorics on these quantities. 
However, the recent article of Ikenmeyer and Panova~\cite{ik-panova:15} showed that 
proving Conjecture~\ref{conj:occ-obstr} in this way is not possible. 
Our work is greatly inspired by~\cite{ik-panova:15}, even though 
it differs vastly in the details and techniques from the present paper
(see \S\ref{se:outline}).

%The recent article 
%proving Conjecture~\ref{conj:occ-obstr} in this way is not possible.  
%Let us point out that while our work was greatly inspired by~\cite{ik-panova:15}, 
%the details and techniques differ vastly.

%In particular, if $k_n(\la)=0$, then $\la$ does not occur in $\C[\Omega_n]$
%This was proposed as a potential method of proving Conjecture~\ref{conj:occ-obstr},
%which has stimulated research in algebraic combinatorics on these quantities. 
%(Note that an occurrence obstruction for 
%$Z_{n,m}\not\subseteq\Omega_n$ 
%is a partition $\la$ for which  $\gctm_n(\la)=0$ 
%and such that $\la$ occurs in $\C[Z_{n,m}]$.) 

Before~\cite{ik-panova:15} there were several papers showing 
that asymptotic considerations cannot resolve Conjecture~\ref{conj:occ-obstr}. 
A first attempt towards finding occurrence obstructions 
was by asymptotic considerations (moment polytopes): 
for Kronecker coefficients 
this was ruled out in~\cite{buci:09}, where it  
was proven that for all $\la\vdash nd$ 
there exists a stretching factor~$\ell\ge 1$ such that $\kron_n(\ell\la) > 0$. 
%Clearly, $\gctm_n(\la)$ is bounded from above by the multiplicity 
%$\plethc_\la(d[n])$ of the partition~$\la$ in the plethysm $\Sym^d\Sym^n\C^{n^2}$. 
Kumar~\cite{Kum:15} 
ruled out asymptotic considerations for 
the GCT-coefficients: assuming the Alon-Tarsi Conjecture~\cite{AT:92}, 
he showed that %Kumar derived that 
$\gctm_n(n\la) > 0$ for all $\la\vdash nd$ and even $n$.
A similar conclusion, unconditional, although with less information on the stretching factor, 
was obtained in~\cite{BHI:15}. 

%The recent article of Ikenmeyer and Panova~\cite{ik-panova:15} showed that 
%proving Conjecture~\ref{conj:occ-obstr} in this way is not possible.  
%put an end to the attempts of succeeding with Conjecture~\ref{conj:occ-obstr} by means of 
%Kronecker coefficients. 

\subsection{Future directions}

While our main result, Theorem~\ref{th:main},  
rules out the possibility of proving 
the Conjecture~\ref{conj:dc-bord} via occurrence obstructions, 
there still remains the possibility that one may succeed so 
by comparing multiplicities. 
If the orbit closure $\ocp_{n,m}$ of the padded permanent
$X_{11}^{n-m} \per_m$ is contained in $\Det_n$, then 
the restriction defines a surjective $G$-equivariant homomorphism 
$\C[\Omega_n]\to\C[\ocp_{n,m}]$ of the coordinate rings, 
and hence  the multiplicity of the type $\la$ in 
$\C[\ocp_{n,m}]$ is bounded from above by
the GCT-coefficient $\gctm_n(\la)$. 
Thus, proving that $\gctm_n(\la)$ is strictly smaller 
than the latter multiplicity implies that  
$\ocp_{n,m}\not\subseteq\Det_n$. 
We note that the paper~\cite{cdw:12} rules out 
one natural asymptotic method for achieving this.
Mulmuley pointed out to us a paper by 
Larsen and Pink~\cite{lapa:90} that is of 
potential interest in this connection. 

%{\tt Structure of paper.}

\section*{Acknowledgments}
{\footnotesize
%Major parts of this research were conducted while CI was at Texas A\&M University and PB was at Paderborn University.
%This research project was initiated in December 2015 when the result \cite{ik-panova:15} was presented at the 1-year reunion workshop of 
%the semester-long program ``Algorithms and Complexity in Algebraic Geometry'' at the Simons Institute for the Theory of Computing in Berkeley, 
%organized by B\"urgisser, Landsberg, Mulmuley, and Sturmfels.
This work, as well as \cite{ik-panova:15}, are an outcome of the program ``Algorithms and Complexity in Algebraic Geometry'', 
held at the Simons Institute for the Theory of Computing in Berkeley in 2014 and organized by B\"urgisser, Landsberg, Mulmuley, and Sturmfels.
We thank all the participants of the program for their interest and stimulating discussions and the Simons Institute for the financial support that 
made this work possible. 
We are grateful to Avi Wigderson for pointing that our proof reveals that, using occurrence obstructions, 
one cannot prove significant lower bounds in the significantly weaker model of power sums. 
}
%%%

\section{Overview of proof}\label{se:outline}

The proof of our main result, Theorem~\ref{th:main}, 
is intricate and combines various new techniques. 
We present below the main ingredients and outline the structure of the proof. 
Moreover, we shall discuss two aspects of our method, that we believe 
are of independent interest: fundamental invariants of forms and the lifting of 
highest weight vectors in plethysms.

\subsection{Outline and ingredients}

The only information we exploit about the orbit closure $\ocp_{n,m}$ of the padded permanent
(recall \eqref{def:Znm}) is an insight due to Kadish and Landsberg~\cite{kadish-landsberg:14}, which 
was also crucially used in Ikenmeyer and Panova~\cite{ik-panova:15}. 
For stating it, we introduce the notion of the {\em body} $\bar\la$ of a partition~$\la$, which 
is obtained be deleting the first component of $\la$, or pictorially, by removing the first row in its diagram. 

\begin{theorem}[\cite{kadish-landsberg:14}]\label{pro:I-KaLa}
If $\la\vdash nd$ occurs in $\IC[Z_{n,m}]_d$, then $\ell(\la)\leq m^2$ and $|\bar\la|\leq md$. 
\end{theorem}

This means that if $\la\vdash nd$ occurs in $\IC[Z_{n,m}]_d$, then $\la$ must have a very long first row
if $n$ is substantially larger than $m$. 
(Note that $|\bar\la|\leq md$ is equivalent to $\la_1\ge (n-m)d$.)

It is remarkable that Proposition~\ref{pro:I-KaLa} still holds if the permanent~$\per_m$ is replaced by 
{\em any} homogeneous polynomial of degree $m$ in $m^2$ variables: 
this can be checked by tracing the proof 
that we provide in Section~\ref{sec:tensorcontrandfuncteval}. 
So no specific property of the permanent is used: only the padding 
(multiplication with a power of a linear form) is relevant. 
Let us point out the recent paper~\cite{GIP:17} which indicates that in other models of computation 
also expressing $\VPs\ne\VNP$, but avoiding the padding (trace of matrix powers), 
no lower bounds can be obtained via occurence obstructions.

For the proof of Theorem~\ref{th:main}, we need to show that many partitions $\la$ occur in $\C[\Det_n]$. 
For this we shall establish the occurrence of certain basic shapes in $\C[\Det_n]$. 
Then we get more shapes by the following \emph{semigroup property}.%
%which states that $\la+\mu$ occurs in $\C[\Det_n]$ if $\la$ and $\mu$ occur in $\C[\Det_n]$
%(see Lemma~\ref{le:SGP}). 

\begin{lemma}\label{le:SGP}
If $\la$ occurs in $\C[\Det_n]$ and $\mu$ occurs in $\C[\Det_n]$,
then $\la+\mu$ occurs in $\C[\Det_n]$.
%If $\la$ and $\mu$ occur in $\C[\Det_n]$, then $\la+\mu$ occurs in $\C[\Det_n]$.
\end{lemma}

\begin{proof}
$\la$ occurs in $\C[\Det_n]$ iff $\C[\Det_n]$ contains a highest weight vector 
of weight $\la$ (see Section~\ref{sec:hwvs} for basics on highest weight vectors). 
Moreover, it is immediate that the product of a highest weight vector of weight~$\la$ 
with a highest weight vector of weight~$\mu$ is a  highest weight vector of weight $\la+\mu$.
Nonzeroness of this product follows from the irreducibility of the variety $\Det_n$.
\end{proof}

%If $f,g\colon W\to\C$ are highest weight vectors with the weights~$\la$ and~$\mu$, 
%respectively, then the product $f\cdot g$ is a highest weight vector of weight $\la+\mu$.
%(See \S\ref{sec:hwvs} for basics on highest weight vectors.)

For the basis shapes it turns out to be sufficient to take rectangular diagrams 
to which there are appended a first column and a (long) first row
(see Theorem~\ref{cor:buildingblock}).  
%The formal statement of the technical result is in 
%(The proof is based on an explicit construction of a highest weight vector.)

In order to explain the strategy in a simpler situation, %For the sake of simplicity, 
let us focus here on the special case of even partitions $\la$, i.e.\ where all parts $\la_i$ are even. 
For the basic shapes, it then suffices to consider rectangular diagrams $k\times \ell$, 
consisting of $k$ rows of length $\ell$, to which a row has been appended. 
Let us call such shapes row extended rectangles. 
The following result provides the basic building blocks. 
We write $\la^{\sharp D}$ for the partition of the ``lifted shape'' $\la+(D-|\la|)$, 
that arises from $\la$ by extending the first row so that $\la^{\sharp D}$ has $D$ boxes.

\begin{proposition}\label{pro:I-explicitHWV}
Let $n \ge k\ell$ and $\ell$ be even. 
Then $(k\times \ell)^{\sharp nk}$ occurs in $\IC[\Det_n]_k$. 
\end{proposition}

%We shall explain the idea of proof of Proposition~\ref{pro:I-explicitHWV}
%in the next subsection.

The strategy of the proof of the Main Theorem~\ref{th:main} is now as follows: 
Suppose we are given an even $\la\vdash nd$ such that $n \ge m^{25}$ and 
$\la$ occurs in $\C[Z_{n,m}]$. 
By Proposition~\ref{pro:I-KaLa} we have 
$\ell(\la)\leq m^2$ and $|\bar\la|\leq md$.

We distinguish two cases. 
If the degree~$d$ is large (say $d \ge 24m^6$) we proceed as in~\cite{ik-panova:15}.  
We decompose the body $\bar\la$ of $\la$ into a sum of even rectangles $k\times \ell$ (i.e., $\ell$ is even). 
Since $n$ and $d$ are sufficiently large in comparison with $m$, 
it turns out that we can write $\la$ as a sum of row extended rectangles $(k \times \ell)^{\sharp nk}$, 
where $n\ge k\ell$.
By Proposition~\ref{pro:I-explicitHWV}, all $(k \times \ell)^{\sharp nk}$ 
occur in $\IC[\Det_n]_k$. The semigroup property then implies that 
$\la$ occurs in $\IC[\Det_n]_d$. 
(See Proposition~\ref{thm:newoccurrence} for more details.) 

If the degree~$d$ is small, we rely on the following result. 
We again assume that $V=\C^{n\times n}$. 

\begin{proposition}\label{pro:I-degreebound}
Let $\la \vdash nd$ be such that there exists a positive integer~$m$ satisfying 
$|\bar{\la}| \le md$ and $md^2\le n$. 
Then {\em every} highest weight vector of weight $\la$ in 
$\Sym^d\Sym^n V$,
viewed as a degree $d$ polynomial function on  $\Sym^n V^*$,
does not vanish on~$\Det_n$. In particular, 
if $\lambda$ occurs in $\Sym^d\Sym^n V$, 
then $\lambda$ occurs in $\IC[\Det_n]_d$.
\end{proposition}

In fact, in order to treat the general case of noneven partitions $\la$, 
we need to make a further case distinction to treat separately 
the case $|\bar\la| < m^{10}$ of a very small body (or extremely long first row).   
See  Proposition~\ref{thm:verylongfirstrow} for details. Fortunately, this case can be 
dealt with by essentially the same techniques as for 
Proposition~\ref{pro:I-degreebound}. 
%We shall explain the ideas underlying this proposition in 
%Subsection~\S\ref{subs:hwv-lift}. 

%For the general case (any $\la$) we substitute Prop.~\ref{pro:I-explicitHWV} by 
%a more complicated result (Thm.~\ref{cor:buildingblock}). 

%We also need to make a further case distinction (extremely small degree $d$, 
%see\ Prop.~\ref{thm:verylongfirstrow}. 

In order to prove the above two propositions, we exploit very little information 
on the orbit closure $\Det_n$ of the determinant. 
The only property we use is that $\Det_n$ contains many padded power sums. 
Here is the formal statement. 

\begin{theorem}\label{eq:I-containlinforms} 
Let $X,\varphi_1,\ldots,\varphi_k$ be linear forms on $\C^{n\times n}$ and assume $n \geq sk$. 
Then the power sum 
$X^{n-s} (\varphi_1^{s}+\cdots+\varphi_k^s)$
of $k$ terms of degree~$s$, padded to degree $n$, is contained in $\Det_n$.
\end{theorem}

\begin{proof} 
The case $n=1$ is trivial, so we assume $k<n^2$ without loss of generality.  
Let $X_1,\ldots,X_{n^2}$ denote the standard basis of $(\C^{n\times n})^*$.
Writing the power sum $X_1^s+\cdots+X_k^s$ as a formula requires at most 
$(s-1)k + k - 1 = sk - 1$ many additions and multiplications.
Valiant's construction~\cite{Val:79a} implies that 
$X_1^s+\cdots+X_k^s$ has the determinantal complexity at most $sk\le n$, 
i.e., it can be written as the determinant of an $n\times n$-matrix 
with affine linear entries in $X_1,\ldots,X_k$. 
By substituting $X_i$ with $X_i/X_{k+1}$ and multiplying with $X_{k+1}^{n}$,  
we see that 
$X_{k+1}^{n-s}(X_1^s+\cdots+X_k^s)$
equals the determinant of an $n\times n$-matrix 
with homogeneous linear entries in $X_1,\ldots,X_{k+1}$. 
It follows that 
$X_{k+1}^{n-s}(X_1^s+\cdots+X_k^s)$ lies in $\Det_n$. 
Let now 
$X,\varphi_1,\ldots,\varphi_k$ be linear forms in $(\C^{n\times n})^*$. 
The assertion follows by 
applying a linear map sending $X_{k+1}$ to $X$ and 
$X_i$ to $\varphi_i$ for $1\le i\le k$. 
\end{proof}

%This theorem follows easily from the fact~\cite{Val:79a} that $\dc(f) \le n$ 
%if $f$ can be written as an arithmetic expression of size $n$. 

We next discuss the ideas underlying the proof of 
Proposition~\ref{pro:I-explicitHWV} and Proposition~\ref{pro:I-degreebound}. 

\subsection{Generic fundamental invariants of tensors}\label{subs:GFI}

%{\tt Give context and explain the proof of Prop.~\ref{pro:I-explicitHWV}.} 

A theorem by Howe~\cite{howe:87} tells us about the smallest degree~$d$ such that 
the plethysm $\Sym^d\Sym^n \C^N$ has nonzero $\SL_N$-invariants: 
$$
 \dim \big(\Sym^d\Sym^n \C^N\big)^{\SL_N} = 
 \left\{
 \begin{array}{ll} 
  0 & \mbox{ if $d<N$,}\\ 
  1 & \mbox{ if $d=N$ and $n$ even,}\\
  0 & \mbox{ if $d=N$ and $n$ odd.}
  \end{array}
 \right.
$$
So if $n$ is even, there is (up to scaling) a unique 
nonzero invariant $F_{n,N}\in \Sym^N\Sym^n \C^N$. 
This invariant was already known to Cayley~\cite{cayley-coll-work} 
and it is sometimes called a ``hyperdeterminant''. We prefer to call it 
the {\em fundamental invariant} of $\Sym^n \C^N$;  
see \cite[eq.~(3.8)]{buik:15} for an explicit formula and generalizations.

Let $V=\C^N$, denote by $e_1,\ldots,e_N$ the standard basis of $V$ 
and by $X_1,\ldots,X_N$ be the basis of $V^*$ dual to the standard basis of $V$.
We interpret $F_{n,N}$ as a polynomial function on the space $\Sym^n V^*$. 
The basic observation is as follows: 
if some $p\in \Sym^n V^*$ satisfies $F_{n,N}(p) \ne 0$, 
then the rectangular partition $d\times n$ 
occurs in the orbit closure of $\ol{\GL_N\cdot p}$. 
Indeed, the restriction of coordinate rings 
$\C[\Sym^n V^*] \to \C[\ol{\GL_N\cdot p}]$ 
maps $F_{n,N}$ to a nonzero function and hence 
the $d\times n$-isotypical component of $\C[\ol{\GL_N\cdot p}]$ is nonzero.

Implementing this idea is not as easy as it may look.
For instance, consider $p=X_1\cdots X_n$ for even~$n$.
Then $F_{n,n}(p) \ne 0$ turns out to be equivalent to the 
Alon-Tarsi conjecture~\cite{AT:92}, which states that 
the number of column-even latin squares of size $n$ is different from 
the number of column-odd latin squares of size $n$
(see \cite{Kum:15,buik:15}). This conjecture is still open.  

However, for the power sum 
$p=c_1 X_1^n +\cdots + c_N X_N^n$ with nonzero coefficients $c_i$,  
one can show that $F_{n,N}(p) \ne 0$; see\ \cite[Thm.~3.19]{buik:15}. 
This is based on a technique introduced in~\cite{bci:10}. 
An extension of this method leads to the proof of Proposition~\ref{pro:I-explicitHWV}, 
which is provided in Section~\ref{sec:BBS}.

\subsection{Lifting of highest weight vectors}\label{subs:hwv-lift} 

The proof of Proposition~\ref{pro:I-degreebound} 
has led us to insights that are of independent interest and that we 
are going to discuss now. 
Consider the plethysm $\Sym^d \Sym^mV$, where again $V=\C^N$.  
The multiplicity of the irreducible $\GL_N$-module of type $\mu\vdash dm$ 
in $\Sym^d \Sym^m V$ equals the dimension of its space 
$\HWV_{\mu}(\Sym^d \Sym^m V)$ of highest weight vectors of weight $\mu\vdash dm$; 
see Section~\ref{sec:hwvs} for these notions. 
The known stability property of plethysms (\cite[Cor.~1.8]{wei:90}, \cite{carre-thibon:92}) 
can be expressed as 
$$
 \dim \HWV_{\mu}(\Sym^d \Sym^mV) = \dim \HWV_{\mu^{\sharp dn}}(\Sym^d \Sym^{n} V) 
$$
for $\mu\vdash md$, $n\ge m$, provided $\mu_2\le m$. 
Recall that $\mu^{\sharp dn}$ denotes the shape obtained from $\mu$ by 
adding $d(n-m)$ boxes to its first row.

We deepen our understanding of this numerical identity by 
constructing an explicit injective linear ``lifting map''  
between the corresponding spaces 
\begin{equation}\label{eq:lift}
  \kappa^d_{m,n}\colon\Sym^d \Sym^m V \to \Sym^d \Sym^{n} V \, , 
\end{equation}
which maps highest weight vectors of weight $\mu\vdash md$ to 
highest weight vectors of weight $\mu^{\sharp dn}$. 
The map $\kappa^d_{m,n}$ arises as the $d$-fold symmetric power of the linear map 
$M\colon\Sym^m V \to \Sym^{n} V,\, p\mapsto p\, e_1^{n-m}$, 
which is the multiplication with $e_1^{n-m}$.

We show how to label a system of generators $v_T$ of $\HWV_{\mu}(\Sym^d \Sym^m V)$ 
by tableaux $T$ of shape~$\mu$ with content $d\times m$ (see Section~\ref{se:sgen}). 
Similar to \cite{bci:10,ike:12b},  we work out a combinatorial rule for contracting~$v_T$ with a tensor of rank one,
which is a technical tool used throughout our developments (see Theorem~\ref{eq:comb_contraction}).  

The lifting of generators turns out to have a beautifully simple combinatorial description in terms 
of the tableaux, see Theorem~\ref{th:lifting-vT}. 
This also leads to a proof of the stability property of plethysms that is  
different from the proofs in \cite{wei:90,carre-thibon:92}.
As a consequence we obtain the following insight, which is crucial for our purposes. 
If $\la\vdash nd$ satisfies $\la_2 + |\bar \la| \le md$, then every 
highest weight vector of weight $\la$ in the right-hand side of \eqref{eq:lift} 
is obtained by lifting a highest weight vector in the left-hand side.

In order to understand the effect of liftings with regard to polynomial evaluation, 
we use duality to show that
$\big\langle\kappa^d_{m,n}(f), q^{d} \big\rangle = \big\langle f, M^*(q)^{d}\big\rangle$ 
for $f\in\Sym^d\Sym^m V$, and $q\in\Sym^n V^*$
(see\ Theorem~\ref{th:lift-ip-preserv}). 
Here $M^*\colon\Sym^n V^* \to \Sym^m V^*$ denotes the dual map of $M$. 
It turns out that $M^*(q)$ essentially 
equals the $(n-m)$-fold partial derivative of $q$ in the direction of $e_1$. 
The proof of Proposition~\ref{pro:I-degreebound}, 
which is given in Section~\ref{sec:SmDeg}, 
combines all the statements mentioned so far, plus 
Proposition~\ref{cor:nonvanishing}.

\subsection{No occurrence obstructions for Waring rank}

%We shed further light on the method of occurrence obstructions by showing that it 
%%An analysis of our proof shows that the method of occurrence obstructions %is weak indeed:
%cannot prove exponential lower bounds on the Waring rank of the permanent, whereas such 
%bounds are known. 

We shed further light on the method of occurrence obstructions by investigating what happens
when replacing $\det_n$ by the sum of $n$th powers of linear forms, thus inquiring about the 
Waring rank of polynomials. The Waring rank is a 
considerably weaker model of computation, which can be seen from the fact 
that for this complexity measure, exponential lower bounds for the permanent can be proven. 
More specifically, the {\em Waring rank} $R(p)$ 
of a polynomial $p\in \Sym^n V^*$ is defined 
as the minimum~$r$ such that there exists a representation 
$p=\varphi^n_1+\ldots + \varphi^n_r$ as a sum of $r$ powers of linear forms $\varphi_i\in V^*$. 
%For  $1 < s < n$, we can think of $\Sym^n V^*$ as being included in 
%$\Sym^s V^* \ot \Sym^{n-s}V^*$. The corresponding matrix (after choosing a basis) is called 
%the {\em catalecticant} (cp.~\cite[Ex.~1.7]{dolgachev:03}).  
%It is easy to see that its rank provides a lower bound on $R(p)$. This way, 
One can show that $R(\det_n)$ and $R(\per_n)$ are both below by ${n\choose \lfloor n/2\rfloor}^2$; 
see~\cite{landsberg-teitler:10}. %for a slightly better bound. 

Before proceeding, we note that $R(X^aY^b) \le b+1$ if $a<b$, which follows from the identity 
\begin{equation*}
 \sum_{j=0}^{b-1} (X+\zeta^j Y)^{a+b} = b X^{a+b} + b {a+b \choose a} X^a Y^b ,
\end{equation*}
where $\zeta$ is a primitive $b$th root of unity. 
(We can even do better: 
if $\omega$ is a primitive $a$th root of unity, the identity 
$$
 \sum_{j=0}^{a-1} (\varepsilon \omega^j X + Y) ^{a+b} = 
  Y^{a+b} + a {a+b \choose a}  X^a Y^b \varepsilon^a + \mathcal{O}(\varepsilon^{a+1}) 
$$
shows that $X^aY^b$ is a limit of polynomials with Waring rank at most $a+1$.) 

One may think of proving lower bounds on the Waring rank by studying the 
orbit closure 
$$
 \PS_n := \ol{\GL_{n^2}\cdot (X^n_1 + \cdots + X^n_{n^2})} \subseteq \Sym^n (\C^{n^2})^* .
$$
of the power sum with $n^2$ terms. Indeed, 
suppose that $p\in\Sym^m V^*$ satisfies $R(p) \le n$ for some $n > 2m$. 
Consider the padded power sum $X^{n-m}p$, where $X\in V^*$. 
It satisfies $R(X^{n-m}p) \le n^2$, since 
we showed $R(X^{n-m} X_i^m) \le n-m+1 \le n$ above. 
Therefore, $X^{n-m}p \in \PS_n$. 
Thus a possible strategy for showing lower bounds on $R(p)$ could be 
to disprove that the padded polynomial $X^{n-m}p$ is contained in $\PS_n$. 
By analogy,  replacing $\Det_n$ with the simpler orbit closure $\PS_n$, 
we may ask whether the separation can be achieved using occurence obstructions.
Unfortunately, the answer turns out to be no!
 
%This orbit closure should have considerably easier geometry than $\Det_n$ 

The only information used about the orbit closure of the determinant~$\Det_n$ in the proof of our main result 
is that it contains certain padded power sums (Theorem~\ref{eq:I-containlinforms}). 
More specifically, for the proof of Theorem~\ref{th:main}, we consider highest weights $\lambda$ occurring in $\C[\Det_n]_d$ 
for different values of $d$ -- Proposition~\ref{pro:I-degreebound} and Proposition~\ref{thm:verylongfirstrow}. 
In both cases we show that each highest weight vector $f \in \Sym^d \Sym^n V$ of weight $\lambda$ of interest appears in $\C[\Det_n]_d$, 
by showing that there is a corresponding power sum $p=\varphi_1^s+\cdots+\varphi^s_k$ on which $f$ does not vanish. 
In Proposition~\ref{pro:I-degreebound} we have $k=d<n, s=md<n$; and in Proposition~\ref{thm:verylongfirstrow} 
we have $k=m^2s, \,  m^2s^2 \leq n$. For all values of $d$ and the corresponding $k$ 
we have that $R(p) \leq k$ and $R(X_1^{n-s}p) \leq n k\leq n^2$, so $X_1^{n-s}p \in \PS_n$ 
and hence $\lambda$ occurs in $\C[\PS_n]_d$.  
%Due to the above observation on the Waring rank of padded power sums, 
Hence we can replace $\Det_n$ by $\PS_n$ and obtain the following result.
%, just by  tracing the proof of Theorem~\ref{th:main}.

\begin{corollary}\label{cor:main-PS}
Let $n,d,m$ be positive integers with 
$n \ge m^{25}$ and $\la\vdash nd$. If $\la$ occurs in 
$\C[Z_{n,m}]$, then $\la$~also occurs in $\C[\PS_n]$. 
\end{corollary}

Recall that in this statement, $Z_{n,m}$ is the orbit closure of the padded permanent 
$X^{n-m} \per_m$; see~\eqref{def:Znm}. 
Tracing the proof reveals that the permanent can be replaced by 
any homogeneous polynomial~$p$ of degree $m$ in $m^2$ variables. 

So we obtain the dramatic result that the strategy of occurrence obstructions 
cannot even be used in the weaker model of $\PS_n$ against padded polynomials.

%As a result, we obtain the dramatic result that the method of occurrence obstructions 
%fails to prove exponential lower bounds on the Waring rank for any polynomial.

%%%
\section{Preliminaries} 
\label{sec:prelim}

\subsection{Symmetric powers}\label{sec:ML} %Multilinear algebra}

We refer to~\cite{northcott:84} for background on multilinear algebra. 
Let us begin with some notational conventions. 
Let $V$ be a finite dimensional complex vector space. 
We write $v_i$ for vectors in $V$, $\varphi_j$ for linear 
forms in $V^*$, and write $\langle v_i, \varphi_j \rangle$ 
or $\langle \varphi_j ,v_i \rangle$ for the contraction $\varphi_j(v_i)$ 
of $v_i$ with $\varphi_j$. 
%We shall also use the following notational conventions:
If $V=\C^N$, then $e_1,\ldots,e_N$ denotes the standard basis of $V$, 
and we denote by $X_1,\ldots,X_N$ the basis of $V^*$ dual to it, so that 
$\langle e_i, X_j\rangle =\delta_{ij}$. 
A basis of the $d$th tensor power $\tensor^d V$ is given by the tensors 
$e_I := e_{i_1} \otimes e_{i_2} \otimes \cdots \otimes e_{i_{d}}$, 
where $I$ runs over all tuples  
$I=(i_1,\ldots,i_d) \in [N]^d$, 
where we write $[N]:=\{1,2,\ldots,N\}$. 
Similarly, we define the dual basis $X_J$ of $\tensor^d V^*$ satisfying 
$\langle e_I, X_J \rangle = \delta_{I,J}$. 

The symmetric group $\aS_d$ on $d$ symbols acts on the $d$th tensor power 
$\tensor^d V$ by permuting the factors. 
The {\em $d$th symmetric power} $\Sym^d V$ 
is the subspace of $\tensor^d V$ consisting of the $\aS_d$-invariant tensors. 
It is obtained as the image of the symmetrizing projection
$$
 \Pi_d :=\frac{1}{d!}\sum_{\pi\in\aS_d} \pi
$$
which we can view as en element of the the group algebra~$\C[\aS_d]$. 
The multiplication maps 
\begin{equation*}\label{eq:symm-prod}
 \Sym^{d_1} V \times \Sym^{d_2} V \to \Sym^{d_1+d_2} V,\, 
 (p,q) \mapsto p\cdot q := \Pi_{d_1+d_2}(p\otimes q) 
\end{equation*}
turn 
$\Sym V := \oplus_{d\in\N} \Sym^d V$ 
into an associative graded $\C$-algebra on which $\GL(V)$ 
acts by algebra automorphisms. 
We note that 
the construction of the symmetric power is functorial: 
a linear map $\Phi\colon V\to W$ of finite dimensional $\C$-vector spaces 
%and $d\in\N$, 
defines a unique linear map 
$\Sym^d\Phi\colon\Sym^d V \to \Sym^d W$ satisfying 
$(\Sym^d\Phi) (v^d) = \Phi(v)^d$. 
Moreover $\Sym^d(\Psi\circ\Phi) = \Sym^d\Psi \circ \Sym^d\Phi$ 
for linear maps $\Psi,\Phi$ that can be composed. 
In particular, the {\em symmetric product} of $d$ vectors $v_1,\ldots,v_d \in V$ is defined as 
$v_1\cdots v_d := \Pi_d(v_1\ot\cdots\ot v_d)$.
%where the symmetrizing projection is given by the element $\Pi_d :=\frac{1}{d!}\sum_{\pi\in\aS_d} \pi$ of the group algebra~$\C[\aS_d]$. 
%Note that $v^d = v^{\ot d}$. 
It is well known that $\Sym^d V$ is spanned by the $d$th powers $v^d= v^{\ot d}$.

We can identify a symmetric tensor $p\in\Sym^d V^*$ with the homogeneous 
polynomial of degree~$d$ given by 
$V\to\C, v\mapsto \langle p,v^{d} \rangle$: one calls 
$p$ the {\em polarization} of the corresponding homogeneous polynomial.  
This way, we can view $\Sym V^*$ as the $\C$-algebra of polynomials in $N$ variables. 
%The multiplication of polynomials corresponds to 
%the {\em symmetric product}  %($d_1,d_2\in\N$)
%\begin{equation}\label{eq:symm-prod}
% \Sym^{d_1} V^* \times \Sym^{d_2} V^* \to \Sym^{d_1+d_2} V^*,\, 
 %(p,q) \mapsto p\cdot q := \Pi_{d_1+d_2}(p\otimes q) .
%\end{equation}
The interpretation of homogeneous degree~$d$ polynomials 
in the variables $X_i$ as symmetric tensors in  $\Sym^d V^*$ 
is highly essential for our work! 

For $\alpha\in\IN^N$ with $|\alpha|=d$ we define the monomial 
$$
 e^\alpha  := e_1^{\alpha_1}\cdots e_N^{\alpha_N} 
 = \Pi_d(e_1^{\ot\alpha_1}\ot\cdots \ot e_N^{\ot\alpha_N})
 = \frac{1}{d!} \sum_{\pi\in\aS_d} \pi e_1^{\ot\alpha_1}\ot\cdots \ot e_N^{\ot\alpha_N}
 = \frac{1}{{d\choose \alpha}} \sum_I e_I ,
$$
where the sum is over all tuples $I \in [N]^d$ with the frequency $\alpha$, i.e.,
$\alpha_k$ is the number of occurences of $k$ in $I$. 
It is well known that the $e^\alpha$ form a basis of $\Sym^d V$. 
Similarly, the $X^\alpha := X_1^{\alpha_1}\cdots X_N^{\alpha_N}$ form a basis 
of $\Sym^d V^*$.
A straightforward calculation shows that 
\begin{equation}\label{eq:dualsym}
\langle e^\alpha, X^\beta \rangle = \frac{1}{{d\choose \alpha}}\, \delta_{\alpha,\beta} ,
\end{equation}
where ${d\choose \alpha}$ denotes the multinomial coefficient.
Therefore, 
$({d\choose \beta} X^\beta)$ is the dual basis of $(e^\alpha)$. 

Let $p\in\Sym^d V^*$ with $d\ge 1$ and $v\in V$.
Viewing $p$ as a polynomial function on $V$, there is a well-defined 
{\em directional derivative} of $p$ at $u\in V$ in direction~$v$: 
$$
 \partial_v p (u) := \lim_{\epsilon \to 0} \frac{1}{\epsilon} \big( p(u +\epsilon v) - p(u)\big) . 
$$
This defines $\partial_v p \in \Sym^{d-1} V^*$.
In coordinates, we have 
$\partial_v p = \sum_i v_i \frac{\partial p}{\partial X_i}$, 
where $v=\sum_i v_i e_i$.
So for fixed $v$, we obtain a linear map
$\partial_v\colon \Sym^d V^*\to\Sym^{d-1}V^*,\, p\mapsto \partial_v p$, 
which we define to be $0$ if $d=0$. 
We denote by $\partial_v^k$ the $k$-fold composition of $\partial_v$. 

%\bigskip

After recalling these general facts, 
we present now a useful lemma on the evaluation of polynomials 
at ``points of low rank''. 

\begin{lemma}\label{le:polar}
Let $W$ be a finite dimensional $\C$-vector space and 
let $p\in\Sym^d W^*$ such that 
$\langle p, (\sum_{j=1}^r w_j)^d \rangle \ne 0$ for some $w_1,\ldots,w_r \in W$.  
Then there exists $J\subseteq [r]$ with $|J| \le d$ and 
$\langle p, (\sum_{j\in J} w_j)^d \rangle \ne 0$. 
\end{lemma}

\begin{proof}
By multilinearity we have 
$$
 0 \ne \Big\langle p, \big(\sum_{j=1}^r w_j \big)^{\otimes d} \Big\rangle = 
 \sum_{j_1,\ldots,j_d} \big\langle p, w_{j_1} \ot \cdots \ot w_{j_d} \big\rangle ,
$$
hence there exist $j_1,\ldots,j_d \in [r]$ such that 
$\langle p, w_{j_1} \ot \cdots \ot w_{j_d} \rangle \ne 0$. 
The polarization formula \cite[p.~5]{dolgachev:03} gives
\begin{equation}\label{eq:pol-formula}
  \big\langle p, w_{j_1}\ot\cdots\ot w_{j_d} \big\rangle = 
  \frac{1}{d!} \sum_{I\subseteq [d]} (-1)^{d-|I|}\,  
    \big\langle p , \big(\sum_{i\in I} w_{j_i} \big)^d\big\rangle  .
\end{equation}
Hence there must be a nonzero contribution for some $I$ 
and the assertion follows.
\end{proof}

%We show now that if $f\in\Sym^d(\Sym^n V)^*$ 
%is nonzero, then $F$ does not vanish on a some power sum 
%$v_1^n + \cdots + v_d^n\in \Sym^n V$ 
%with at most $d$ terms, where $v_i \in V$. 
%More specifically: 

\begin{proposition}\label{cor:nonvanishing}
Let $V$ be a finite dimensional $\C$-vector space and $d,n\ge 1$. 
If $f\in\Sym^d\Sym^n V$ is nonzero, then 
$\langle f, (\varphi_1^n + \cdots + \varphi_d^n)^d \rangle \ne 0$ 
for Zariski almost all $(\varphi_1,\ldots,\varphi_d) \in (V^*)^d$. 
This means that~$f$, viewed as a homogeneous polynomial function of degree~$d$ 
on $\Sym^n V^*$, does not vanish on almost all power sums $\varphi_1^n + \cdots + \varphi_d^n$ 
with $d$ terms.
\end{proposition}

\begin{proof} 
Let $w\in \Sym^n V^*$ be such that $\langle f, w^d \rangle \ne 0$. 
There exist $r\in\N$ and 
$\psi_1,\ldots,\psi_r \in V^*$ such that 
$w=\psi_1^n +\ldots + \psi_r^n$. By Lemma~\ref{le:polar}, there are 
$1\le j_1,\ldots,j_d \le r$ such that 
$\langle f, (\psi_{j_1}^n +\ldots + \psi_{j_d}^n)^d \rangle \ne 0$. 
Therefore, the open set 
$\{(\varphi_1,\ldots,\varphi_d) \in (V^*)^d \mid \langle f, (\varphi_1^n +\ldots +\varphi_d^n)^d \rangle\ne 0 \}$ 
is nonempty and hence Zariski dense. 
\end{proof}

%%%
%\subsection{Padded power sum embedding}\label{sec:ppsembedding}

%Let $V=\C^{n\times n}$ and $\Det_n\subseteq\Sym^n V^*$  
%denote the orbit closure of~$\det_n$; see~\eqref{def:Omegan}. 
%The following result goes back to Valiant~\cite{Val:79a}. 
%Remarkably, it is the only property of the orbit closure~$\Det_n$ that we exploit in our paper!

%%%
\subsection{Basics on highest weight vectors}\label{sec:hwvs}

We refer to \cite{FH:91} for more details and proofs for the following basic notions and facts.

A {\em partition} $\la$ is a nondecreasing finite sequence of nonnegative integers $(\la_1,\ldots,\la_N)$. 
It can be visualized as a {\em Young diagram}, 
which is a finite collection of boxes, arranged in left-justified rows, 
with $\la_i$ boxes in the $i$th row. Depending on the context, 
we also write $\la$ for the set of boxes of the diagram. 
%For convenience, we recall notations already introduced previously:
Recall that 
$|\la| := \sum_i \la_i$ is the {\em size} of $\la$, 
which is the number of the boxes of the Young diagram, 
and $\ell(\la)$ is the {\em length} of~$\la$, which is defined as 
the number of its nonzero parts~$\la_i$.
%, i.e., the number of rows of the diagram.
%We call $\la_i$ the \emph{length} of the $i$th row of $\la$.
We briefly write $\la \vdash D$ to express that $\la$ is a partition of size $D$. 
The {\em body} $\bar\la$ of $\la$ is obtained from $\la$
by removing its first row. 
We denote by 
$k\times\ell$ the rectangular diagram with $k$ rows 
of length $\ell$, so it has size $k\ell$. 
In partition notation, 
$k\times\ell = (\ell,\ldots,\ell)$ with $\ell$ appearing $k$ times. 

Recall that $\aS_D$ denotes the symmetric group on $D$ symbols. 
It is well known that the irreducible $\aS_D$-modules can be encoded 
by partitions $\la\vdash D$ of size~$D$. They are called {\em Specht modules}  
and we shall denote them by $[\la]$. 

%Let $G:=\GL_N(\IC)$ and 

We denote by $U_N\subseteq \GL_N(\C)$ the subgroup of upper triangular matrices with ones on the main diagonal.
Moreover, let $\diag(\alpha_1,\ldots,\alpha_N)$ denote the diagonal matrix with entries $\alpha_i$
on the diagonal.
Suppose that $\sV$ is a rational $\GL_N(\C)$-module. 
%be a finite dimensional $\C$-vector space and a rational $\GL_N(\C)$-module. 
A nonzero vector $f \in \sV$ is called a \emph{highest weight vector} of weight $\la\in\IZ^N$ iff 
$f$ is $U_N$-invariant, i.e., $u \cdot f = f$ for all $u \in U_N$, 
and $f$ is a weight vector of weight $\la$, i.e.,
$\diag(\alpha_1,\ldots,\alpha_N)\cdot f = \alpha_1^{\la_1} \cdots \alpha_N^{\la_N} f$ 
for all $\alpha_i \in\IC^\times$. 
We remark that necessarily $\la_1\ge\ldots\ge\la_N$, 
so that $\la$ is a partition if its entries are nonnegative. 
We denote by $\HWV_\la(\sV)$ the vector space of highest weight vectors of weight~$\la$.
An irreducible  $\GL_N$-module $\sV$ is called a {\em Schur-Weyl module}.
%and we shall denote it $\{\la\}$.
It is known that there is a unique $\la$ such that 
$\HWV_\la(\sV)$ is one-dimensional. Moreover, $\HWV_\mu(\sV)=0$ for all $\mu\ne\la$. 
We call $\la$ the {\em type} of the Schur-Weyl module $\sV$
and shall abbreviate $\sV$ by the symbol~$\{\la\}$. 
Isomorphic $\GL_N$-modules have the same type.

In the following we assume $V:=\IC^N$ and denote by $e_1,\ldots,e_N$ the standard basis of $V$. 
The group $\GL(V)$ acts on the $D$th tensor power~$\tensor^D V$  
by $g (v_1\ot\cdots\ot v_D) = (gv_1)\ot\cdots\ot (gv_D)$ and the group $\aS_D$ 
acts by permuting the factors. Since these actions commute, we have an action of 
$\GL(V) \times \aS_D$ on $\tensor^D V$.  

We next explain how to construct highest weight vectors in $\tensor^d V$.
%We denote by $X_1, \ldots, X_{N}$ the standard basis vectors of $\IC^N$.
Let $\la\vdash D$ and $\mu$ denote the transpose of $\la$, 
so $\mu_i$ denotes the number of boxes in the $i$-th column of $\la$.
For $j \le N$ we note that 
$v_{j\times 1} := e_1 \wedge e_2 \wedge \cdots \wedge e_j \in \bigwedge^j V$ 
is a highest weight vector of weight $j \times 1$. 
We define now: 
\begin{equation}\label{eq:vla}
v_\la := v_{\mu_1 \times 1} \ot\cdots\ot v_{\mu_{\lambda_1} \times 1} \ \in\ \tensor^D V.
\end{equation}
It is easy to check that $v_\la$ is a highest weight vector of weight $\la$. 

\begin{proposition}\label{pro:orbit-HWV} 
Let $\la\vdash D$. Then the vector space $\HWV_\la(\tensor^D V)$ is spanned by 
the $\aS_D$-orbit of~$v_\la$. 
\end{proposition}

\begin{proof}
Schur-Weyl duality provides a $\GL(V) \times \aS_D$-isomorphism 
\[
\tensor^D V \simeq \bigoplus_{\la \vdash D} \{\la\} \otimes [\la] .
\]
Recalling that $\HWV_\la(\{\la\})$ is one-dimensional, we see that 
$\HWV_\la(\tensor^D V)$ is isomorphic to $[\la]$ as an $\aS_D$-module. 
It follows that $\HWV_\la(\tensor^D V)$ is spanned by the $\aS_D$-orbit of any of its nonzero elements. 
\end{proof}

We analyze now the tensors $v_\la$ in more detail. 
A {\em Young tableau} of shape $\la\vdash D$ 
is a filling of the boxes of the diagram $\la$ with numbers. 
We shall assume that each of the numbers $1,\ldots,D$ occurs exactly once, 
so that we obtain an enumeration of the boxes. 
%We consider numberings of the boxes of the diagram $\la$ with the numbers $1,\ldots,d$. 
The \emph{column-standard Young tableau $T^\textup{std}_\la$ of shape~$\la$} is the Young tableau 
of shape~$\la$ that contains the numbers $1,\ldots,D$ ordered columnwise, 
from top to bottom and left to right. For example,
\[
\Yvcentermath1 T^\textup{std}_{(4,2)} = \young(1356,24)
\]
is column-standard.
The symmetric group $\aS_{D}$ acts on the set of Young tableaux %numberings of the boxes 
of the diagram~$\la$ by replacing each entry $i$ with $\pi(i)$.
For example, for $\pi= (2453)$, we obtain
\begin{equation}\label{eq:exatableau}
\Yvcentermath1 \pi T^\textup{std}_{(4,2)} = \young(1236,45).
\end{equation}

Describing $\pi v_\la$ in terms of the permutation~$\pi$ has redundancies 
that can be read off the tableau $\pi T^\textup{std}_{\la}$. Namely, the following 
is an easy consequence of the definition of $v_\la$ in~\eqref{eq:vla}. 
For a transposition $\tau=(i \ j)$ we have: 
%if $i$ and $j$ are in the same column in $\pi T^\textup{std}_{\la}$, then
\begin{equation}\label{eq:transpositioninverse}
 \tau \pi v_\la = -\pi v_\la \quad\mbox{if $i$ and $j$ are in the same column in $\pi T^\textup{std}_{\la}$.}
\end{equation}
%for the transposition $\tau=(i \ j)$.
Moreover, if $\tau$ is a permutation that switches two columns 
of the same length in $\pi T^\textup{std}_{\la}$, then 
\begin{equation}\label{eq:exchangecols}
 \tau \pi v_\la = \pi v_\la .
\end{equation}

\section{Plethysms}\label{se:plethysms}

%Again assume $V=\IC^N$. 
We partition the position set 
$[dn] := \{1,\ldots,dn\}$
into the {\em blocks} $B_1,\ldots,B_d$, 
where 
$B_u := \{ (u-1)n +v \mid 1\le v \le n\}$. 
The subgroup of $\aS_{dn}$ of permutations 
that preserve the partition into blocks
is called the {\em wreath product} $\aS_d \wr \aS_n$.
It is generated generated by the permutations 
leaving the blocks invariant,  
and the permutations of the form 
$(u-1)n+v \mapsto (\tau(u) -1) n +v$ with $\tau \in\aS_d$, 
which simultaneously permute the blocks. Structurally, 
the wreath product is a semidirect product 
$\aS_d \wr \aS_n \simeq(\aS_n)^d \rtimes  \aS_d$. 
Note that its order equals $d!\, n!^d$. 
Symmetrizing over $\aS_\out \wr \aS_\inn$, we obtain 
\begin{equation}\label{eq:PROJ}
\Sigma_{d,n} := \frac{1}{d!\, n!^d} \sum_{\sigma \in \aS_d \wr \aS_n} \sigma 
\end{equation}
in the group algebra $\C[\aS_{dn}]$. 
We obtain the plethysm $\Sym^d \Sym^n V$ as the space of 
$\aS_d \wr \aS_n$-invariants in $\tensor^{dn} V$. 
This space is the image of the projection 
$\tensor^{dn} V \twoheadrightarrow \Sym^d\Sym^n V, w\mapsto \Sigma_{d,n} w$ 
induced by~$\Sigma_{d,n}$. 
We define the {\em plethysm coefficient}, for $\la\vdash dn$, 
\begin{equation*}%\label{eq:def-pleth-coeff} 
 a_\la(d[n]) :=\dim\HWV_\la(\Sym^\out \Sym^\inn V) ,
\end{equation*}
as the multiplicity of $\{\la\}$ in $\Sym^\out \Sym^\inn V$. 
%(If $\la\vdash dn$ violated, then the multiplicity is zero.) 

\subsection{A system of generators encoded by tableaux}\label{se:sgen}
%Highest weight vectors in plethysms}%\label{se:plethysms}

Again let $V=\IC^N$ with the standard basis $e_1,\ldots,e_N$. 
As a consequence of Proposition~\ref{pro:orbit-HWV}, 
the highest weight vector space of $\Sym^\out \Sym^\inn V$ 
of weight $\la\vdash dn$ is spanned 
by the projections $\Sigma_{d,n} \pi v_\la$, where $\pi$ 
runs over all permutations in $\aS_{dn}$. 
However, this description is highly redundant: 
we have 
$\Sigma_{d,n}\pi = \Sigma_{d,n}\pi'$ iff 
$(\aS_d \wr \aS_n)\pi = (\aS_d \wr \aS_n)\pi'$,
for $\pi,\pi'\in\aS_{dn}$.
We next give an intuitive description of the cosets of $\aS_d \wr \aS_n$ 
in terms of certain tableaux.

\begin{definition}\label{def:tab-rect-cont}
A {\em tableau $T$ of shape $\la\vdash dn$ with %rectangular 
content $\out \times \inn$} is a partition %$\la= C_1 \cup\ldots \cup C_d$ 
of the set of boxes of the Young diagram of $\la$  
into $d$ classes $C_1,\ldots,C_d$, each of size $n$. %such that $|C_u| = n$ for all~$u$.
\end{definition}

Intuitively, we think of a such a tableau $T$ as a filling of the Young diagram of $\la$ with 
$d$ different letters, where all boxes in the class~$C_u$ have the same letter  
(and the order of the letters is irrelevant). 

We assign to a permutation $\pi\in\aS_{dn}$ and $\la\vdash dn$ 
a tableau~$T_\la(\pi)$ of shape $\la$ with content $d\times n$ as follows: 
take $d$ different letters and replace in $\pi T^{\textup{std}}_\la$ the numbers in the block~$B_u$ 
by the same letter. 
It is clear that any tableau of shape $\la\vdash dn$ with content $d\times n$
can be obtained this way. 
For example, for the tableau in~\eqref{eq:exatableau} we get for $n=3,d=2$, 
using the letters $a,b$, 
\[
\Yvcentermath1 \young(aaab,bb) = \young(bbba,aa).
\]
It should be clear that $T_\la(\pi) = T_\la(\pi')$ iff 
$(\aS_\inn \wr \aS_\out)\pi = (\aS_\inn \wr \aS_\out)\pi'$, 
for $\pi,\pi'\in\aS_{dn}$. 

By this observation, the following is well-defined.

\begin{definition}\label{def:v_T}
Let $T$ be a tableau of shape $\la$ with content $\out \times \inn$.
We define 
$v_T := \Sigma_{d,n}\pi v_\lambda$
where $\pi\in\aS_{dn}$ is such that $T = T_\la(\pi)$. 
\end{definition}

By Proposition~\ref{pro:orbit-HWV}, $v_T$ is a highest weight vector
in $\Sym^d\Sym^n V$ of weight $\la$. 
We can restrict our attention to certain $T$ because of the following. 

\begin{lemma}\phantomsection\label{cla:wedgevanish}
\begin{enumerate}
\item Let $T$ be a tableau of shape $\la$ with content $\out \times \inn$.
If the same letter appears in a column of $T$ more than once, then $v_T=0$.
\item Let $T$ and $T'$ be two tableaux of shape $\la$ with content $\out \times \inn$
that can be obtained from each other by switching two columns that have the same length.
Then $v_T=v_{T'}$.
\end{enumerate}
\end{lemma}

\begin{proof}
For the first assertion let $(r,c)$ and $(r',c)$ be different positions in~$T$ in the same column 
that have the same letter.
Assume $T=T_\la(\pi)$ and let $i$ and $j$ denote the entries of $\pi T^{\textup{std}}_\la$ 
at the positions $(r,c)$ and $(r',c)$, respectively.
Then $i$ and $j$ lie in the same block, since they are mapped to the same letter. 
Hence the transposition $\tau := (i \ j)$ is an element of $\aS_d \wr \aS_n$.
%which is the key observation.  
Using~\eqref{eq:transpositioninverse}, we see that
symmetrizing $\pi v_\la$ over the 2-element subgroup 
$\{\text{id}, \tau\} \subseteq \aS_d \wr \aS_n$ maps $\pi v_\la$ to zero. 
Hence symmetrizing over the full group $\aS_{d} \wr \aS_n$ maps $\pi v_\la$ 
to zero as well. 

The second assertion is shown analogously, 
but using \eqref{eq:exchangecols} instead of \eqref{eq:transpositioninverse}.
\end{proof}

By a \emph{singleton column} we understand a column of length one. 
According to Lemma~\ref{cla:wedgevanish}(2), singleton columns of $T$ 
can be permuted without changing the value of $v_T$. 
Moreover, according to Lemma~\ref{cla:wedgevanish}(1), 
we can restrict attention to tableaux,  
where no letter appears more than once in a column. 
This leads to the following definition. 

\begin{definition}\label{def:tabl-class}
Let $T$ and $T'$ be tableaux of shape $\la$ with content $d \times n$, 
where no letter appears more than once in a column.
We call $T$ and $T'$ {\em equivalent} if they differ only by a reordering of their singleton columns.
In this case we write $T \simeq T'$.
\end{definition}

By Lemma~\ref{cla:wedgevanish}(2), $v_T$ depends only on the equivalence class of $T$.
For example, the following two tableaux with content $2\times 3$ are equivalent: 
\[
\Yvcentermath1 \young(aaab,bb) \simeq
 \young(aaba,bb) .
\]

Summarizing, we arrived at the following result.

\begin{proposition}\label{pro:HWVSS-gen}
The vector space $\HWV_\la(\Sym^\out \Sym^\inn V)$ is spanned by the 
highest weight vectors~$v_T$, where $T$ ranges over all equivalence classes of tableaux 
of shape $\la$ with content $\out \times \inn$,  
such that no letter appears more than once in a column of $T$.
\end{proposition}

%%%
\subsection{Explicit contractions} %Contracting highest weight vectors in plethyms with rank one tensors}
\label{sec:tensorcontrandfuncteval}
%Evaluation of highest weight vectors in plethysms} 

%{\tt The upshot of this section is mainly an explicit calculation.}   

Our goal here is to work out an explicit combinatorial rule for contracting a highest weight vector in the plethysm
with a tensor of rank one. This will be our main technical tool for proving that a specific highest weight vector 
does not vanish on the orbit closure $\Det_n$ or $Z_{n,m}$, respectively.  

Again we assume $V=\C^N$ and denote by 
$X_1, \ldots, X_{N}$ the basis of $V^*$ dual to the standard basis $e_1,\ldots,e_N$ of $V$. 
The basis vectors $X_{s(1)} \ot\ldots\ot X_{s(dn)}$ of $(V^*)^{\ot dn}$ 
are encoded by maps $s\colon [dn] \to [N]$. 
%Our goal is to provide a combinatorial description of the contraction 
%$\langle v_T, X_{s(1)} \ot\cdots\ot X_{s(dn)} \rangle$.
%where $\langle\ ,\ \rangle$ denotes the standard inner product. 
We view $\la\vdash dn$ as a Young diagram %and, for convenience,
so that $\la$ also denotes the set of boxes of this diagram.
Recall that a tableau $T$ of shape~$\la$ with content $d\times n$ is 
given by a partition $\la =C_1 \cup\ldots \cup C_d$
of the set of boxes of $\la$ into classes $C_u$ of 
size $n$.
Also, recall the decomposition 
$[dn]=B_1\cup\ldots \cup B_d$ into blocks of size~$n$ 
introduced at the very beginning of Section~\ref{se:plethysms}. 
We consider now bijective assignments $\vartheta\colon\la \to [dn]$ 
that map boxes in the same class to numbers in the same block. 
More formally: 

\begin{definition}\label{def:respects}
Let $T$ be a tableau of shape $\la$ with content $d\times n$ 
and $\vartheta\colon\la \to [dn]$ be a bijection. 
We say that $\vartheta$ {\em respects $T$} 
iff there exists a permutation $\tau\in\aS_d$ such that 
$\vartheta(C_i) = B_{\tau(i)}$ for all~$i$. 
\end{definition}

This notion is closely related to the wreath product $\aS_d\wr\aS_n$ as follows. 
Let us call  {\em standard enumeration} $\vartheta_0\colon\la \to [dn]$ %which respects $T_\la(\id)$. 
the labeling of the boxes in $T^\textup{std}_\la$ and let $\pi\in\aS_{dn}$. 
Then the assignments~$\vartheta$ respecting $T_\la(\pi)$ are given by 
$\vartheta = \sigma\circ\pi\circ\vartheta_0$, where $\sigma\in \aS_d\wr\aS_n$.  

It is useful to introduce some further notations. 
Let $j=(j_1,\ldots,j_k)$ be a list of integers. 
If $\{j_1,\ldots,j_k\}=\{1,2,\ldots,k\}$, then $\sgn(j)$ 
denotes the {\em sign} of the permutation $j$; otherwise, we 
define $\sgn(j) :=0$. 
For instance, $\sgn(2,1,3)=-1$ and $\sgn(2,1,2)=0$. 

Suppose $\vartheta\colon\la \to [dn]$ respects the tableau~$T$ of shape~$\la$ 
with content $d\times n$, and take a map $s\colon [dn] \to [N]$. 
We define the {\em value} $\val_\vartheta(s)$ of~$\vartheta$ 
at $s\colon [dn] \to [N]$ by 
\begin{equation}\label{eq:prodcols}
 \val_\vartheta (s) := \prod_{\textup{column $c$ of } \la} \sgn (s\circ\vartheta)_{|_c} .
\end{equation}
Here, $\sgn (s\circ\vartheta)_{|_c}$ denotes the sign of the list of integers 
$(s(\vartheta(1,c)),\ldots,s(\vartheta(\mu_c,c))$ 
corresponding to the $c$th column of the diagram $\la$. 
It is important to note that if $\val_\vartheta (s) \ne 0$, then 
$s(\vartheta(\Box))=1$ for all singleton columns $\Box$ of $\la$.

We shall use the following %important 
rule throughout the paper.

\begin{theorem}\label{eq:comb_contraction} 
Let $T$ be a tableau of shape $\la\vdash dn$ with content $d\times n$
and $s\colon [dn] \to [N]$ be a map. %an assignment. 
Then 
\begin{equation*}
\langle v_T, X_{s(1)} \ot\ldots\ot X_{s(dn)} \rangle = \frac{1}{d!\, n!^d} 
 \sum_{\vartheta} \val_\vartheta (s) ,
\end{equation*}
where the sum is over all bijections $\vartheta\colon\la\to [dn]$ respecting $T$. 
\end{theorem}

\begin{proof} 
We first note that for $\varphi_1,\ldots,\varphi_k \in V^*$, the contraction
$\langle e_1\wedge\ldots\wedge e_{k},  \varphi_{1} \ot\ldots\ot \varphi_{k} \rangle$ 
equals the determinant of the top $k\times k$ minor 
of the $N\times k$ matrix whose columns are $\varphi_{1},\ldots,\varphi_{k}$.
In particular, if $\varphi_i = X_{s(i)}$ are from the dual standard basis, we get 
$$
\langle e_1\wedge\ldots\wedge e_{k},  X_{s(1)} \ot\ldots\ot X_{s(k)} \rangle = 
 \sgn (s(1),\ldots,s(k)).
$$
%this determinant is either zero or equals the sign of the corresponding permutation.
Let $\mu$ denote the transpose of $\la$ and put $r:=\la_1$.
From the definition~\eqref{eq:vla} of the highest weight vector~$v_\la \in \tensor^{dn} V$, we get 
%contracting $v_\la$ with $t$ yields the following product of determinants 
$$
 \langle v_\la , X_{s(1)} \ot\ldots\ot X_{s(dn)} \rangle = 
 \langle e_1\wedge\ldots\wedge e_{\mu_1},  X_{s(1)} \ot\ldots\ot X_{s(\mu(1))} \rangle \cdots 
 \langle e_1\wedge\ldots\wedge e_{\mu_r},  X_{s(dn-\mu_r+1)} \ot\ldots\ot X_{s(dn)} \rangle 
$$
Combined with the above observation we obtain that
\begin{equation}\label{eq:wedgerule}
 \langle v_\la , X_{s(1)} \ot\ldots\ot X_{s(dn)} \rangle = \val_{\vartheta_0} (s) ,
\end{equation}
where $\vartheta_0$ is the standard enumeration of $\la$. 

Suppose that $T=T_\la(\pi)$ for $\pi\in\aS_{dn}$. 
%, i.e., %the tableau $T$ with content $d\times n$ $T$ results from 
%$\pi T^\textup{std}_\la$ by substituting the numbers in the block $B_u$ by the letter~$u$. 
%By Definition~\ref{def:v_T} we have $v_T= \Sigma_{d,n} \pi v_\la$.  
From the definition of $\Sigma_{d,n}$ in~\eqref{eq:PROJ} and 
$v_T= \Sigma_{d,n} \pi v_\la$, %(Definition~\ref{def:v_T}) 
we obtain for any $\Phi\in \Sym^d\Sym^n V^*$ that 
\begin{equation*}%\label{eq:fomel}
 \langle v_T,\Phi\rangle %   \langle \Sigma_{d,n} \pi v_\lambda,\Phi\rangle = 
 =  \frac{1}{d!\, n!^d} \sum_{\sigma \in \aS_d\wr \aS_n} \langle \sigma\pi v_\la , \Phi\rangle
 = \frac{1}{d!\, n!^d} \sum_{\sigma \in \aS_d \wr \aS_n} \langle v_\la , (\sigma\pi)^{-1} \Phi\rangle.
\end{equation*}
Applying \eqref{eq:wedgerule} to 
$\Phi = X_{s(1)} \ot\ldots\ot X_{s(dn)}$, we obtain 
$$
 \langle v_\la , (\sigma\pi)^{-1} (X_{s(1)} \ot\ldots\ot X_{s(dn)}) \rangle 
 = \langle v_\la , X_{s\sigma\pi(1)} \ot\ldots\ot X_{s\sigma\pi(dn)}) \rangle 
 = \val_{\vartheta_0} (s\sigma\pi) .
$$ 
Recall that if $\sigma$ runs through the wreath group $\aS_d\wr\aS_n$, 
then $\vartheta = \sigma\circ\pi\circ\vartheta_0$  
runs through the assignments respecting $T_\la(\pi)$. 
Using this, we easily see that 
$\val_{\vartheta_0} (s\sigma\pi) = \val_{\vartheta} (s)$, 
which completes the proof.
\end{proof}

As a first application we give a short proof of the following result from~\cite[Thm.~3.19]{buik:15} 
on the evaluation of highest weight vectors on power sums. 
An extension of this will lead us later to the proof of Proposition~\ref{pro:I-explicitHWV}. 

%{\tt This should be emphasized more! 
%Put a remark explaining that the corresponding contraction for 
%$X_1\cdots X_m$ leads to the Alon-Tarsi conjecture. Cite  Diss Ikenmeyer and Kumar. 
%Also point out results in \cite{buik:15} and the possibly the connection to
%Cayley's invariant and its generalization. Maybe some of this should show up 
%in the introduction, so that the connection to these old and difficult classical questions 
%becomes apparent.}

\begin{corollary}\label{cor:ps-eval}
Let $n$ be even and $T$ be the tableau of shape $d\times n$ with content $d\times n$, 
in which in each row all boxes have the same letter. Then 
$\big\langle v_T, \big(c_1 X_1^n+\ldots+c_dX_d^n \big)^{d} \big\rangle = d! \, c_1\ldots c_d$
for $c_1,\ldots,c_d\in \C$. 
\end{corollary}

\begin{proof} 
%Put $w:=c_1 X_1^n+\ldots+c_dX_d^n$. 
%Since $X_i^n=X_i^{\otimes n}$ by~\eqref{def:monomial}, 
We have 
$(c_1 X_1^n+\ldots+c_dX_d^n)^{\ot d} = %w^{\otimes d} = 
\sum_{i_1,\ldots,i_d} c_{i_1}\cdots c_{i_d} X_{i_1}^{\otimes n} \otimes \ldots \otimes X_{i_d}^{\otimes n}$, 
where the sum is over all $1\le i_1,\ldots,i_d\le d$. 

Assume first $(i_1,\ldots,i_d)=(1,\ldots,d)$. 
We apply Theorem~\ref{eq:comb_contraction} to compute 
$\langle v_T, X_{1}^{\otimes n} \otimes \cdots \otimes X_{d}^{\otimes n} \rangle$.
Using a rowwise enumeration of the boxes of $\la$, the bijections
$\vartheta\colon\la\to [dn]$ respecting $T$
are in one to one correspondence with the elements of the wreath product 
$\aS_d\wr \aS_n$. They are given by permutations $\tau\in\aS_d$ 
of the rows, and permutations of the numbers within the rows. 
Such a bijection $\vartheta$ contributes 
$\val_\vartheta(s)=\sgn(\tau)^n$, 
where $s=(1,\ldots,1,\ldots,d,\ldots,d)$
(each index occuring $n$ times). 
Hence we obtain, by the assumption that $n$ is even,  
$$
\langle v_T, c_1\cdots c_d X_{1}^{\otimes n} \otimes \cdots \otimes X_{d}^{\otimes n} \rangle
 = \frac{c_1\cdots c_d}{d!} \sum_{\tau\in\aS_d} \sgn(\tau)^n = c_1\cdots c_d .
$$
%This equals $c_1\cdots c_d$ if $n$ is even, and $0$ if $n$ is odd and $d>1$. 
We turn now to the contributions of an arbitrary sequence $(i_1,\ldots,i_d)$. 
If it is a permutation of $1,\ldots,d$, then, by the same argument as before, 
we see that we get the same contribution. 
On the other hand, if the sequence is not a permutation of $1,\ldots,d$, 
we get zero. 
This proves the assertion.
\end{proof}

As a further application, we prove the following general result, which immediately
implies Theorem~\ref{pro:I-KaLa}.
%which is due to Kadish and Landsberg~\cite{kadish-landsberg:14}. 
%For the sake of clarity 

\begin{theorem}\label{thm:KaLaG}
Let $V$ be a finite dimensional $\C$-vector space and $n\in\N$. 
\begin{enumerate}
\item Assume $\pi\colon V\to W$ is a projection,  
$p\in\Sym^n W^*$, and let $Z$ denote the $\GL(V)$-orbit closure of $\pi^*(p)\in\Sym^n V^*$. 
If $\la\vdash nd$ occurs in $\IC[Z]_d$, then $\ell(\la)\leq \dim W$.

\item Let $q=X^{n-m}p$ where $X\in V^*$ and $p\in\Sym^mV^*$ for $m\le n$.
Let $Z$ denote the $\GL(V)$-orbit closure of $q$. 
If $\la\vdash nd$ occurs in $\IC[Z]_d$, then $|\bar\la|\leq md$. 
\end{enumerate}
\end{theorem}

\begin{proof}
(1) For the first assertion, choose a basis $e_1,\ldots,e_N$ of $V$ such that $e_1,\ldots,e_M$ is a basis of~$W$. 
Then $W^*$ is the span of $X_1^*,\ldots,X_M^*$, where $X_1,\ldots,X_N$ denotes the dual basis of $V^*$. 
Assume $\ell(\la) > M$. We need to prove that $\la\vdash dn$ does not occur in $\C[Z]$. 
By Proposition~\ref{pro:HWVSS-gen}, this means that 
$\langle v_T, \pi^*(p)^{\ot d} \rangle = 0$
for all tableau~$T$ of shape~$\la$ with content $d\times n$.
For this, it is enough to show that $\langle v_T, \Phi\rangle =0$ for all tensors 
$\Phi=X_{s(1)}\otimes\cdots\otimes X_{s(dn)}$, where 
$s\colon [dn] \to [M]$.
The first column $c$ of $\la$ has $\ell$ boxes. 
We compose $s$ with a bijection $\vartheta\colon\la \to [dn]$.
Then the restriction of 
$s\circ\vartheta$ to the first column $c$ 
is not injective since $\ell>M$. 
Hence $\sgn (s\circ\vartheta)_{|_c} =0$, 
and Theorem~\ref{eq:comb_contraction} implies that indeed 
$\langle v_T, \Phi\rangle =0$.

(2) We now choose the basis $(e_i)$ of $V$ such that $X=X_1$ is 
the first element of the dual basis of~$(e_i)$. 
Assume $|\bar\la|>md$, so that $\la_1 < (n-m)d$. 
We need to prove that 
$\langle v_T, q^{\ot d} \rangle = 0$
for all tableau~$T$ of shape~$\la$ with content $d\times n$.
We can express $q^{\otimes d}$ as a linear combination of tensors
$\Phi=X_{s(1)}\otimes\cdots\otimes X_{s(dn)}$, where 
$s\colon[dn]\to [N]$ maps at least $(n-m)d$ elements to $1$.
Fix such a tensor~$\Phi$ and consider a bijection $\vartheta\colon\la\to [dn]$.  
Since $\la_1 < (n-m)d$, $\la$ has less than $(n-m)d$ columns. 
By the pigeonhole principle, there is a column~$c$ 
in which $s\circ\vartheta$ puts a $1$ in at least two boxes. 
Therefore, $\sgn (s\circ\vartheta)_{|_c} =0$ and 
Theorem~\ref{eq:comb_contraction} implies that indeed
$\langle v_T, \Phi\rangle =0$.

\end{proof}

%%%
\section{Lifting highest weight vectors in plethysms}\label{sec:kllifting}

%{\tt This section is of independent mathematical interest.} 

In the following we analyze two ways of ``lifting'' highest weight vectors 
in $\Sym^d\Sym^m V$ by raising either the {\em inner degree~$m$} 
or the {\em outer degree}~$d$. 
If $d$ and $m$ are sufficiently large in comparison with~$|\bar\mu|$
for $\mu \vdash dm$, 
then these liftings provide isomorphisms of the 
spaces of highest weight vectors. In particular, the multiplicity~$a_\mu(d[m])$ does not increase, 
which is known as the stability property of the plethysm coefficients
\cite{wei:90,carre-thibon:92,mani:98}. 
A detailed understanding of the lifting in terms of highest weight vectors 
and the tableaux encoding them 
is central for the proof of our main result. 
%and this is not considered in \cite{wei:90,carre-thibon:92,mani:98}. 
We believe the result of this section are of independent mathematical interest. 
As a side result, we also obtain new proofs for the stability properties. 

%{\tt The following should be streamlined. 
%The inner lifting map 
%\begin{equation}%\label{eq:in-deg-lift}
%\kappa^d_{m,n}\colon \Sym^d \Sym^m V \to \Sym^d \Sym^n V
%\end{equation}
%should appear at the beginning, along with its properties. 
%More intrinsic understanding of this map?} 

%%%
\subsection{The multiplication maps and their duals}

Again assume $V=\C^N$ with the standard basis $e_1,\ldots,e_N$. 
Let $n\ge m$ and consider the multiplication with $e_1^{n-m}$ % (see~\eqref{eq:symm-prod})
\begin{equation}\label{eq:Me1}
 M\colon\Sym^m V \to \Sym^n V,\, p \mapsto p e_1^{n-m}  . %= \Pi_{n}(p\ot e_1^{n-m}) .
\end{equation}
Clearly, this is an injective linear map. 
Moreover, $M$ is $U_N$-equivariant, 
since $U_N$ acts on $\Sym V$ via algebra automorphisms and leaves $e_1$ invariant.

%Indeed, the projection~$\Pi_n$ commutes with the action of $\GL(V)$ 
%(since $\Pi_n\in \C[\aS_n]$), hence we have for $u\in U_N$ and $p\in\Sym^mV$,
%using $u\cdot e_1 = e_1$:
%\begin{equation}\label{eq:U-equivar}
% u \cdot M(p) %= u \Pi_n (e_1 \ot p) = \Pi_n u(e_1\ot p)  
% = \Pi_n( u\cdot (p \ot e_1^{n-m}))  = \Pi_n ( (u\cdot p) \ot (u\cdot e_1)^{n-m}) 
% = M(u\cdot p) .
%\end{equation}
%%using $u\cdot e_1 = e_1$. 
%%$M$ maps the highest weight vector $e_1^m$ of $\Sym^m V$ 
%%to the highest weight vector $e_1^n$ of $\Sym^n V$.  

We determine now the dual $M^*\colon\Sym^n V^* \to \Sym^m V^*$
of the linear map~$M$. It turns out to be proportional to the $(n-m)$-fold directional 
derivative $\partial_{e_1}^{n-m}$ in the direction $e_1$, see~Section~\ref{sec:ML}.
%up to a normalization factor.  

\begin{lemma}\label{eq:dualM}
For $q\in\Sym^n V^*$ we have 
$M^*(q)=  \frac{m!}{n!}\, \partial_{e_1}^{n-m}q$. 
\end{lemma}
%For all $f\in\Sym^mV$ and $q\in\Sym^n V^*$ we have 
%$M^* = \frac{m!}{n!} \partial^{n-m}_{e_1}$, which means 
%$$
% \langle M(f), q \rangle = \frac{m!}{n!}\, \langle f, \partial_{e_1}^{n-m}q \rangle .
%$$

\begin{proof}
 We have to show that 
for all $f\in\Sym^mV$ and $q\in\Sym^n V^*$: 
%$M^* = \frac{m!}{n!} \partial^{n-m}_{e_1}$, which means 
$$
 \langle M(f), q \rangle = \frac{m!}{n!}\, \langle f, \partial_{e_1}^{n-m}q \rangle .
$$
By bilinearity, it suffices to check the equality for the basis elements 
$f=e^\alpha$ and $q=X^\beta$, where $|\alpha|=m$ and $|\beta|=n$. 
We put $\delta:=(n-m,0,\ldots,0)$. 
By~\eqref{eq:dualsym} we have 
\begin{equation}\label{eq:eX}
 \langle e^{\alpha+\delta}, X^\beta \rangle = \frac{1}{{n\choose \beta}} \delta_{\alpha+\delta,\beta}.
\end{equation} 
We assume $\beta_1 \ge n-m$, otherwise this expression vanishes. We calculate 
$$
 \partial_{e_1}^{n-m} q = \beta_1 ( \beta_1-1)\cdots ( \beta_1-n+m+1) X^{\beta-\delta} 
$$
and obtain 
$$
 \frac{m!}{n!}\,\langle e^{\alpha}, \partial_{e_1}^{n-m} q \rangle 
  = \frac{m!}{n!}\,\beta_1 ( \beta_1-1)\cdots ( \beta_1-n+m+1) 
  \frac{1}{{m\choose \alpha}} \delta_{\alpha,\beta-\delta} . 
$$ 
It is easily verified that this equals \eqref{eq:eX}.
\end{proof}

The map $M^*$ is clearly surjective. We now show that when restricting it to the subspace 
$X_1^{n-m}\Sym^m V^*$, we get an isomorphism.
%We note the following easily proven fact, which is useful for computing~$M^*$. 

\begin{lemma}\label{le:newscaling}
The map 
$\Delta_{m,n}\colon\Sym^m V^* \to \Sym^m V^*,\, p\mapsto M^*(X_1^{n-m}p)$
is a linear automorphism.
\end{lemma}

\begin{proof}
We expand the polynomial 
$p=\sum_{i=0}^m a_i X_1^i$ with respect to the variable $X_1$. It is easy to check that   
$$
 \frac{\partial^{n-m}}{\partial X_1^{n-m}} (X_1^k p) = \sum_{i=0}^m (i+n-m)(i+n-m-1)\cdots (i+1) a_i X_1^i ,
$$
hence $p\mapsto \frac{\partial^{n-m}}{\partial X_1^{n-m}} (X_1^{n-m} p)$ is injective. 
Now we apply Lemma~\ref{eq:dualM}. 
\end{proof}

%%%
\subsection{Inner degree lifting}

We denote by $\kappa^d_{m,n} := \Sym^dM$ the $d$-fold symmetric product  
of the map 
$M\colon\Sym^m V \to \Sym^n V$,
and call this the {\em inner degree lifting} by $n-m$:
\begin{equation}\label{eq:SMe1}
 \kappa^d_{m,n}\colon \Sym^d\Sym^m V \to \Sym^d\Sym^n V,\, p^d \mapsto M(p)^d 
 = (p\, e_1^{n-m})^d .
\end{equation}
Clearly, $\kappa^d_{m,n}$ is an injective linear map. 
This map behaves nicely regarding highest weight vectors. 

\begin{lemma}\label{le:LNEW}
The inner lifting $\kappa^d_{m,n}$ maps 
highest weight vectors of weight $\mu\vdash dm$ 
to highest weight vectors of weight $\mu^{\sharp dn}$. %$\mu+(d(n-m))$.
\end{lemma}

\begin{proof}
The $U_N$-equivariance of $M$ immediately extends to its $d$-fold symmetric product $\kappa^d_{m,n}$. 
Therefore, $\kappa^d_{m,n}$ maps $U_N$-invariants to $U_N$-invariants. 

To conclude the proof, it suffices to check that $\kappa^d_{m,n}$ 
raises the weight of weight vectors by $(d(n-m),0,\ldots,0)$. 
For this, we note that $\kappa^d_{m,n}$ is obtained by restriction from 
the $d$th tensor power 
$\tensor^d M\colon \tensor^d\Sym^m V \to \tensor^d\Sym^n V$ of $M$, 
so that it suffices to verify this property for $\tensor^d M$.

Let $\Sym^m V = \oplus_\alpha \C e^\alpha$ be the weight space decomposition, where 
the sum is over all $\alpha\in\N^m$ and 
$e^\alpha := e_1^{\alpha_1}\cdots e_N^{\alpha_N}$
has the weight $\alpha$.  
We obtain the decomposition
$$
 \tensor^d \Sym^m V = \bigoplus_{\alpha^1,\ldots,\alpha^d} \C\, 
 e^{\alpha^1}\ot \cdots \ot e^{\alpha^d} ,
$$
where $e^{\alpha^1}\ot \cdots \ot e^{\alpha^d}$ has the weight $\alpha^1+\cdots+\alpha^d$.
The multiplication map $M$ sends $e^{\alpha}$ to 
$e^{\alpha+ \delta}$, where $\delta :=(n-m,0,\ldots,0)$.
Hence the tensor power $\tensor^d M$ of $M$ sends 
$e^{\alpha^1}\ot \cdots \ot e^{\alpha^d}$ to 
$e^{\alpha^1+\delta}\ot \cdots \ot e^{\alpha^d+\delta}$, 
which has the weight $\alpha^1+\cdots+\alpha^d+ d\delta$. 
\end{proof}

The dual of the inner degree lifting
$$
 (\kappa^d_{m,n})^*\colon \Sym^d\Sym^n V^* \to \Sym^d\Sym^m V^*
$$
it easily described: since $\kappa^d_{m,n} =\Sym^d M$, we have 
$(\kappa^d_{m,n})^* =\Sym^d M^*$. 
This immediately implies the following important observation that we 
state as a theorem. 

\begin{theorem}\label{th:lift-ip-preserv}
Let $n\ge m$, $f\in\Sym^d\Sym^m V$, and $q\in\Sym^n V^*$. Then
$$
\big\langle \kappa^d_{m,n}(f), q^{d} \big\rangle = \big\langle f, M^*(q)^{d}\big\rangle .
$$ 
\end{theorem}

This allows to express the evaluation of the lifted $\kappa^d_{m,n}(f)$ at~$q$, 
viewed as a polynomial function on $\Sym^n V^*$, as the evaluation of $f$ at $M^*(q)$.
Recall that $M^*(q)$ equals (up to a normalizing factor) the $(n-m)$-fold directional derivative of~$q$ 
in direction $e_1$ (see Lemma~\ref{eq:dualM}). 
(We remark that 
in the conference version of our paper from FOCS 2016 \cite{BIP:16}, there is an error 
regarding the definition of the lifting map and the statement of Theorem~\ref{th:lift-ip-preserv}.)

Recall the generators $v_T$ of $\Sym^d\Sym^m V$ labeled by tableaux $T$
(Definition~\ref{def:v_T}).  
We show now that $\kappa^d_{m,n}$ %the inner degree lifting 
maps $v_T$ to a generator $v_{T'}$, whose tableau $T'$ arises from $T$ 
in a simple way. 

\begin{theorem}\label{th:lifting-vT}
Let $T$ be a tableau of shape $\mu$ with content $d\times m$ and 
let the tableau $T'$ of shape $\mu':=\mu^{\sharp dn}$ %\mu+(d(n-m))$ 
with content $d\times n$ 
be obtained from $T$ by adding $n-m$ copies of each of the $d$ letters 
in the first row (in some order). 
Then $\kappa^d_{m,n} (v_T) = v_{T'}$.  
\end{theorem}

\begin{proof}
Since we can decompose $\kappa^d_{m,n} = \kappa^d_{n-1,n}\circ\cdots\circ\kappa^d_{m,m+1}$
into liftings by one, it suffices to prove the assertion for those. 
We thus focus on $\kappa := \kappa^d_{m,m+1}$. 
%in which case $\mu' := \mu + (d)$. 

Suppose that $T=T_\mu(\pi)$ for $\pi\in\aS_{dm}$ so that 
$v_T = \Sigma_{d,m}(\pi v_\mu)$ (see Definition~\ref{def:v_T}).
By definition, the tableau $T'$ of shape $\mu' := \mu + (d)$ is 
obtained from $T$ by adding one copy of each of the $d$ letters 
in the first row (in some order). 
Let us first determine a permutation $\pi' \in \aS_{d(m+1)}$ such that 
$T' \simeq T_{\mu'}(\pi')$ (see Definition~\ref{def:tabl-class}).
We denote by $\rho\in\aS_{d(m+1)}$ the permutation that merges 
the last $d$ entries into the first $d$ blocks, each at the end of the block, respectively. 
More specifically,
$$
 \rho(v_1\ot\cdots \ot v_{dm}\ot w_1\ot\cdots\ot w_d) = 
  v_1\ot\cdots \ot v_{m}\ot w_1\ot v_{m+1}\cdots\ot v_{2m}\ot w_2\ot\cdots \ot v_{(d-1)m+1}\ot\cdots\ot v_{dm}\ot w_d .
$$
(For example, for $m=3,d=2$, we have $\rho = (4567)$.) 
%$\rho^{-1}=\begin{pmatrix}
%1&2&3&4&5&6&7&8\\
%1&2&3&7&4&5&6&8
%\end{pmatrix}$ and hence 
%$\rho = (4567)$. 
We view $\pi\in\aS_{dm}\subseteq\aS_{d(m+1)}$ and define now 
$\pi' := \rho \pi\in\aS_{d(m+1)}$.
%(thus $\tilde{\pi}$ acts like $\pi$ on the first $dm$ factors.)
The reader should verify that we indeed have 
$T' \simeq T_{\mu'}(\pi')$. %(see Definition~\ref{def:tabl-class}).
%Therefore,
%$v_{T'} = \Sigma_{d,m+1} (\tensor^d \Pi_{m+1}) (v_{\mu'})$.

We have $v_{\mu'} = v_\mu \ot e_1^{\ot d}$ and hence 
\begin{equation}\label{eq:rhof}
 \pi' v_{\mu'} = \rho \pi v_{\mu'} 
 = \rho \pi (v_\mu \ot e_1^{\ot d}) = \rho ((\pi v_\mu)\ot e_1^{\ot d}) .
\end{equation}

We can write the symmetrizer $\Sigma_{d,m}$  
of the wreath group $\aS_d\wr\aS_m$ (see~\eqref{eq:PROJ})
as the composition 
$\Sigma_{d,m} = \Pi_{d,m}\circ \tensor^d \Pi_m$,
where 
$\Pi_{d,m}$ denotes the symmetrizer of the subgroup isomorphic to $\aS_d$, 
consisting of the permutations in $\aS_{dn}$ of the form
$(u-1)n+v \mapsto (\tau(u) -1) n +v$ with $\tau \in\aS_d$, 
(simultaneous permutations of the blocks). 
We define 
$$
 w_T := \big(\tensor^d \Pi_m\big) (\pi v_\mu), \quad 
 w_{T'} := \big(\tensor^d \Pi_{m+1}\big) (\pi' v_{\mu'}) 
$$
so that we can write $v_T = \Pi_{d,m}(w_T)$ and $v_{T'} = \Pi_{d,m+1}(w_{T'})$. 

Let $M\colon\Sym^m V \to \Sym^{m+1} V,\, p \mapsto p\, e_1 = \Pi_{n}(p \ot e_1)$ 
denote the multiplication by $e_1$ (see~\eqref{eq:Me1})
and note that the following diagram is commutative: 
$$
\begin{array}{cccc}
\tensor^d\Sym^m V & \stackrel{\bigotimes^d M}{\longrightarrow} & \tensor^d\Sym^{m+1} V \\
  \scriptstyle{\Pi_{d,m}}\ \downarrow  &                   & \downarrow\ \scriptstyle{\Pi_{d,m+1}}   \\
\Sym^d\Sym^m V & \stackrel{\Sym^d M}{\longrightarrow} & \Sym^d\Sym^{m+1} V .
\end{array}
$$
Therefore, it is sufficient to prove that 
\begin{equation}\label{eq:wT}
 \big(\tensor^d M\big) (w_T) = w_{T'} .
\end{equation}

For showing this,
we first note that for $p\in\tensor^m V$ we have 
\begin{equation}\label{le:MM}
 \Pi_{m+1}(p \ot e_1) = \Pi_{m+1}(\Pi_m(p) \ot e_1) = M (\Pi_m(p)) .
\end{equation}
We claim now that for $w\in\tensor^d\tensor^m V$, we have
\begin{equation}\label{eq:MClaim}
 (\tensor^d \Pi_{m+1}) \big(\rho ( w \ot e_1^{\ot d}) \big) = \big(\tensor^d (M\circ \Pi_m)\big) (w) .
\end{equation}
For verifying this, we may assume that 
$w=w_1\ot\cdots\ot w_d$ with $w_i\in\tensor^m V$. 
By the definition of~$\rho$, we have 
$\rho(w \ot e_1^{\ot d}) = w_1\ot e_1 \ot \cdots\ot w_d\ot e_1$,  
and hence, using~\eqref{le:MM} for the second equality, 
$$
 \big(\tensor^d\Pi_{m+1}\big) (\rho (w\ot e_1^{\otimes d})) = \tensor_{i=1}^d \Pi_{m+1} (w_i\ot e_1)
 = \tensor_{i=1}^d M (\Pi_m(w_i)) %= \big(\tensor_{i=1}^d (M \circ \Pi_m) \big) (w) ,
$$
and \eqref{eq:MClaim} follows.

Using~\eqref{eq:rhof} and~\eqref{eq:MClaim} with $w=\pi v_\mu$, 
we argue now as follows:  
$$
 w_{T'} = \big(\tensor^d \Pi_{m+1} \big) (\pi' v_{\mu'}) =
 \big(\tensor^d \Pi_{m+1} \big) \rho \big( (\pi v_\mu) \ot e_1^{\ot d} \big) = 
 \big(\tensor^d (M\circ \Pi_m)\big)(\pi v_\mu) = 
 \big(\tensor^d M\big)(w_T) ,
$$
which shows \eqref{eq:wT} and completes the proof. 
\end{proof}

Theorem~\ref{th:lifting-vT} now easily implies the stability property of plethysms 
with respect to inner degree lifting. 

\begin{proposition}\phantomsection\label{pro:lifting}
\begin{enumerate}
\item Suppose $\mu\vdash md$ is such that $\mu_2 \le m$ and let $n\ge m$. 
Then the inner degree lifting $\kappa^d_{m,n}$ %\eqref{eq:SMe1} 
defines an isomorphism
\begin{equation*}\label{eq:InnerLift}
\HWV_{\mu}(\Sym^\out \Sym^\imm V) \to \HWV_{\mu^{\sharp dn}}(\Sym^\out \Sym^{n} V) ,\ 
 f\mapsto \kappa^d_{m,n}(f) .
\end{equation*}

\item Suppose that $\la\vdash nd$ satisfies $\la_2\le m$ and $\la_2 + |\bar{\la}| \le md$. 
Then every highest weight vector of weight~$\la$ in 
$\Sym^d\Sym^n V$  is obtained by lifting a highest weight vector in $\Sym^d\Sym^{m} V$ 
of weight $\mu$, where $\mu\vdash md$ such that $\bar{\mu} = \bar{\la}$.
\end{enumerate}
\end{proposition}

\begin{proof}
(1) We need to show the surjectivity of the map. 
Let $T'$ be a tableau of shape $\mu^{\sharp dn}$ with content $d\times n$
such that no letter appears more than once in a column. 
Then each of the $d$ letters appears at least 
$n-\mu_2\ge n-m$ times in singleton columns. 
Hence 
$T'$ is obtained from a tableau~$T$ of shape $\mu$ with content $d\times m$ 
as in Theorem~\ref{th:lifting-vT} (the order of the letters in the singleton columns is irrelevant, 
see Definition~\ref{def:tabl-class}). 
Since $\kappa^d_{m,n}(v_T) = v_{T'}$ by Theorem~\ref{th:lifting-vT},
the assertion follows with Proposition~\ref{pro:HWVSS-gen}.

(2) Note that $\la_2 + |\bar{\la}| \le md$ is the number of boxes of $\la$ that appear in columns 
that are not singleton columns. We can therefore shorten the given~$\la$ to a partition 
$\mu\vdash md$ by removing singleton columns. 
Then $\bar{\mu}= \bar{\la}$ and 
$\la = \mu^{\sharp dn}$ and we conclude with part one. 
\end{proof}

%We remark that the stability property of plethysms 
%expressed in Proposition~\ref{pro:lifting}(1) 
%was first shown in \cite[Cor.~1.8]{wei:90} and \cite{carre-thibon:92}
%by different methods
%(not relying on the construction of highest weight vectors).  

%%%
%\subsection{A fundamental contraction property of the inner degree lifting}

%%%
\subsection{Outer degree lifting} 
\label{se:deglift}

We keep the notation from before.
%$V=\C^N$ with the standard basis $e_1,\ldots,e_N$ of $V$ 
%and the dual basis $X_1,\ldots,X_N$ of $V^*$ from the previous section. 
%Note that $e_1^m\in \Sym^m V$ is a highest weight vector of weight $(m)$. 
Let $k\le d$. By multiplying with the $(d-k)$th power of $e_1^m$, 
we obtain the injective linear map
\begin{equation}\label{eq:deglift}
\Sym^k \Sym^m V \to \Sym^d \Sym^m V,\, 
 f\mapsto (e_1^m)^{d-k} \cdot f .
\end{equation}
Since $\GL(V)$ acts on $\Sym V$ by algebra automorphisms, 
it follows that \eqref{eq:deglift} maps a highest weight vector 
of weight $\nu\vdash mk$ to a highest weight vector of weight 
$\nu^{\sharp dm}$. 

%{\tt Note somewhere that: for $f\in\Sym^k\Sym^m V$ and $p\in\Sym^m V^*$ 
%\begin{equation}\label{eq:outdeg-triv} 
% \langle (e_1^m)^{d-k} \cdot f , p^d \rangle = \langle e_1^m,p \rangle^{d-k} \cdot \langle f, p\rangle^k . 
%\end{equation}}

\begin{lemma}\label{le:deglift}
Let $T$ be a tableau of shape $\nu$ with content $k\times m$  
and let the tableau $T''$ of shape $\nu^{\sharp dm}$ be obtained from $T$ 
by adding $m$ copies of $d-k$ new letters to the first row (in some order). 
Then $v_{T''}$ is obtained as the image of $v_T$ under the map~\eqref{eq:deglift}
\end{lemma}

\begin{proof} %{\tt Simplify?} 
The lifting from degree $k$ to~$d$ can be obtained as a composition of 
liftings that increase the degree by one only.  Hence we may assume without loss of generality that $d=k+1$. 
So we assume that $T''$ is obtained from $T$ by adding $m$ copies of a new letter to the first row. 
In order to show that $v_{T''} = e_1^m\cdot v_T$, it is suffices to prove that 
\begin{equation}\label{eq:vT''-eval}
 \langle v_{T''}, \Phi \rangle = \langle e_1^m \cdot v_T, \Phi\rangle 
\end{equation}
for all tensors $\Phi=\varphi_1\otimes \cdots \otimes \varphi_{k+1}$,
where $\varphi_i\in \otimes^m V^*$ is of rank one.  
We shall evaluate the inner products with Theorem~\ref{eq:comb_contraction}. 
For this, consider the decomposition
$[(k+1)m] = B_1 \cup\ldots\cup B_{k+1}$ into the blocks~$B_i$ 
of size $m$ as in the beginning of Section~\ref{se:plethysms}. 
For $1\le i\le k+1$, let $\tau_i\in\aS_{(k+1)m}$ be the permutation that 
exchanges the blocks~$B_i$ and~$B_{k+1}$ and preserves the order 
within the blocks (note $\tau_{k+1}=\id$). 

Suppose that $\vartheta\colon\nu\to [km]$ respects~$T$ and 
put $\nu'':=\nu^{\sharp dm}$.
We extend $\vartheta$ to a map 
$\tilde{\vartheta}\colon\nu''\to [km+m]$ 
by sending the $m$ boxes of $T''\setminus T$ bijectively to the
numbers in the block $B_{k+1}$.
Clearly $\tilde\vartheta$ respects $T''$. 
Let $1\le i \le k+1$ and $\pi$ be a permutation of $B_{k+1}$.  
Composing~$\tilde{\vartheta}$ with $\tau_i \pi$,  
we obtain a map 
$\vartheta''\colon\nu''\to [km+m]$ 
respecting $T''$. Moreover, any assignment $\vartheta''$ respecting $T''$ 
arises this way from uniquely determined $\vartheta$, $i$, and $\pi$. 

Taking this observation into account, we deduce from 
Theorem~\ref{eq:comb_contraction} after some thought that: 
\begin{equation*}
\begin{split}
 \langle v_{T''}, \Phi \rangle &= 
  \frac{1}{k+1}\Big(
 \langle v_T, \varphi_1\ot\cdots\ot \varphi_k \rangle \cdot \langle  e_1^m, \varphi_{k+1}\rangle  \\
  &+ \sum_{i=1}^{k} \langle v_T, \varphi_1\ot\cdots\ot \varphi_{i-1}\ot \varphi_{k+1}\ot \varphi_{i+1}\cdots\ot\cdots\ot \varphi_k \rangle 
    \cdot \langle  e_1^m, \varphi_{i}\rangle \Big) .
\end{split}
\end{equation*}
(Note that the first summand corresponds to $i=k+1$.)
This equals $\langle e_1^m\cdot v_T , \Phi \rangle$ 
by the definition of the symmetric product. 
\end{proof}

The outer degree lifting~\eqref{eq:deglift} behaves nicely with respect to highest weight vectors. 

\begin{proposition}\phantomsection\label{le:newlift}
\begin{enumerate}
\item Suppose $\nu\vdash mk$ such that 
$\nu_2 + |\bar{\nu}|\le k$ and let $d\ge k$. 
The lifting \eqref{eq:deglift} defines an isomorphism
\begin{equation*}%\label{eq:outer-lifting}
 \HWV_{\nu}(\Sym^k \Sym^m V) \to \HWV_{\nu^{\sharp dm}}(\Sym^d \Sym^m V) .
\end{equation*}

\item Suppose that $f$ is a highest weight vector in $\Sym^d\Sym^m V$ of weight 
$\mu\vdash dm$ and assume that $\mu_2 +|\bar{\mu}| \le k\le d$ for some $k$. 
Then $\mu=\nu^{\sharp dm}$ for some $\nu\vdash mk$ and 
$f = (e_1^m)^{d-k} \cdot g$ for some $g\in \HWV_{\nu}(\Sym^k\Sym^m V)$. 
 
\end{enumerate}
\end{proposition}

\begin{proof} 
(1) We need to show the surjectivity. 
Let $T''$ be a tableau of shape $\mu :=\nu^{\sharp dm}$ with content $d\times m$. 
Note that $\nu_2 + |\bar{\nu}|$ is the number of boxes in $\mu$ that are not singleton columns.
Hence there are least
$d - (\nu_2 + |\bar{\nu}|) \ge d - k$ many letters appearing in singleton columns of $T''$ only.   
Removing the $(d-k)m$ many boxes with these letters from the first row of $\mu$ leads to 
a tableau $T$ of shape $\nu$ with content $k\times m$. We conclude with 
Lemma~\ref{le:deglift} and Proposition~\ref{pro:HWVSS-gen}.

(2) There exists a tableau $T''$ of shape $\mu$ with content $d\times m$ 
by assumption and Proposition~\ref{pro:HWVSS-gen}.
As before, we see that there are at least $d-k$ many letters 
appearing in singleton columns of $T''$ only.  
In particular, we have $\mu =\nu^{\sharp dm}$ for some $\nu\vdash mk$. 
Now we apply part one.
\end{proof}

We remark that the stability of plethysm in Proposition~\ref{le:newlift}(1),
for the slightly weaker condition $|\bar{\nu}| \le k \le d$,
was first shown in~\cite{mani:98} with a geometric method. 

\section{Proof of Theorem~\ref{th:main}}

%%%
\subsection{Small degrees or extremely long first rows}\label{sec:SmDeg}

In the following we put $V:=\C^{n\times n}$. 
To warm up, we first show that $(n)$ occurs in $\IC[\Det_n]_1$.
Indeed, we have $X_{1}^n \in\Det_n$ and $e_1^n \in \Sym^n V$ is 
a highest weight vector of weight $(n)$ such that 
$\langle e_1^n, X_1^n \rangle = 1$.
%The linear map $\Sym^n V^* \to\C,\, p\mapsto \langle e_1^n,p\rangle$
%is a highest weight vector in $\Sym^n V$ of weight $(n)$; 
%moreover it maps $X_{1}^{n}$ to $1$. 
%Since $X_{1}^n \in\Det_n$, we see that $(n)$ indeed occurs in $\IC[\Det_n]_1$. 

%The following result deals with the case of partitions $\la \vdash nd$ where $d$ is small.

%\begin{proposition}\label{pro:degreebound}
%Let $\la \vdash nd$ be such that there exists a positive integer~$m$ satisfying 
%$|\bar{\la}| \le md$ and $md^2\le n$. 
%Then every highest weight vector~$h$ of weight $\la$ in 
%$\Sym^d\Sym^n V$,
%viewed as a degree $d$ polynomial function on $\Sym^n V^*$,
%does not vanish on~$\Det_n$. In particular, 
%if $\lambda$ occurs in $\Sym^d\Sym^n V$, 
%then $\lambda$ occurs in $\IC[\Det_n]_d$.
%\end{proposition}

\begin{proof}[Proof of Proposition~\ref{pro:I-degreebound}]
We assume that $\la \vdash nd$ is such that there exists a positive integer~$m$ satisfying 
$|\bar{\la}| \le md$ and $md^2\le n$. 
Further, let $h\in\Sym^d\Sym^n V$ be any highest weight vector of weight $\la$.
We need to prove that there exists $q\in\Det_n$ such that
$\langle h,q^d \rangle \ne 0$. 
The case $d=1$ is trivial as noted before. So suppose $d\ge 2$. 

We have $\la_2 \le |\bar{\la}| \le md$ and 
$\la_2 + |\bar{\la}| \le 2 |\bar{\la}| \le 2md \le md\cdot d$.
Therefore, we are in the setting of Proposition~\ref{pro:lifting}(2) 
with respect to the lifting (note that the inner degree on the left-hand side is $md$ and not $d$) 
$$
 \Sym^d\Sym^{md} V \to \Sym^d\Sym^{n} V .
$$ 
We conclude that $h$ arises by an inner degree lifting from 
a highest weight vector~$f\in\Sym^d\Sym^{md} V$ of weight $\la$; 
so we have $h=\kappa^d_{md,n}(f)$. 
Recall the linear automorphism  
$\Delta:=\Delta_{md,n}\colon\Sym^{md} V^* \to \Sym^{md} V^*$ 
from Lemma~\ref{le:newscaling}. 

We view $f$ as a degree~$d$ homogeneous polynomial map $\Sym^{md} V^* \to \C$ and 
apply Proposition~\ref{cor:nonvanishing} to the composition 
$ \Sym^{md} V^* \to \C,\, p\mapsto \langle f, \Delta(p)^d\rangle$.
Hence there is a power sum 
$$
 p = \varphi_1^{md}+\cdots+\varphi_d^{md} 
$$
with at most $d$ terms 
such that $\langle f, \Delta(p)^d \rangle \ne 0$, 
for some $\varphi_1,\ldots,\varphi_{d}\in V^*$. 
We apply now Theorem~\ref{th:lift-ip-preserv} with  $q:=X_1^{n-md}\, p$  
and obtain
$$
  \langle h, q^{d} \rangle = \langle \kappa_{m,n}^d(f), q^{d} \rangle = 
  \langle f, M^*(q)^{d} \rangle =\langle f, \Delta(p)^{d} \rangle \ne 0 .
$$
(Note that $M^*(q) = \Delta(p)$ by the definition of $\Delta$, see~Lemma~\ref{le:newscaling}.) 
% = X_1^{n-md}\, \big(\varphi_1^{md}+\cdots+\varphi_d^{md}\big)$, 
By Theorem~\ref{eq:I-containlinforms} we have 
$q\in \Det_n$ since $n\ge md\cdot d$. 
The assertion follows.
\end{proof}

While Proposition~\ref{pro:I-degreebound} is meant to deal with the case 
of partitions $\la \vdash nd$ where $d$ is small, the 
next result deals with the extreme situation, where the body $\bar{\la}$ is very small compared 
with the size of $\la$. 
(In this situation, the splitting strategy in the proof of Proposition~\ref{thm:newoccurrence} 
below would fail.) 

\begin{proposition}\label{thm:verylongfirstrow}
Let $\la\vdash nd$ and assume there exist positive integers $s,m$ such that 
$\ell(\la) \le m^2$, $\la_2 \le s$, $m^2s^2 \le n$, and $m^2 s \le d$. 
Then every highest weight vector $h\in\Sym^d\Sym^n V$ of weight $\la$, 
viewed as a degree $d$ polynomial function on  $\Sym^n V^*$,
does not vanish on~$\Det_n$.   
\end{proposition}

\begin{proof}
We first consider the inner degree lifting 
$\Sym^d\Sym^s V \to \Sym^d\Sym^n V$, 
see~\eqref{eq:SMe1}. Since  
$$ 
\la_2\le s, \quad
\la_2 + |\bar{\la}| \le \la_2 + (\ell(\la)-1) \la_2 =  \ell(\la) \la_2 
 \le m^2 s \le d \le ds,
$$  
the assumptions of Proposition~\ref{pro:lifting}(2) 
are satisfied and we conclude that $h$ arises by 
lifting some $f\in\HWV_{\mu}(\Sym^d\Sym^{s} V)$ 
with $\mu\vdash ds$ and $\bar{\mu} = \bar{\la}$. 

By assumption, $k:= m^2 s \le d$. 
We continue with the outer degree lifting map 
$\Sym^k\Sym^s V \to \Sym^d\Sym^s V$, see~\eqref{eq:deglift}.
We have, using the above,  
$$
 \mu_2 + |\bar{\mu}| = \la_2 + |\bar{\la}| \le m^2 s  = k , 
$$ 
hence the assumptions of Proposition~\ref{le:newlift}(2) 
are satisfied and we have $f= (e_1^s)^{d-k}\cdot g$ 
for some highest weight vector $g\in\Sym^{k}\Sym^s V$ of 
weight~$\nu\vdash ks$ such that
$\bar{\nu}=\bar{\mu}$. 

Recall the linear automorphism  
$\Delta:=\Delta_{s,n}\colon\Sym^{s} V^* \to \Sym^{s} V^*$ 
from Lemma~\ref{le:newscaling}
and consider power sums 
$p = \varphi_1^{s} + \cdots + \varphi_{k}^{s}$.
By Proposition~\ref{cor:nonvanishing}, there are $\varphi_1,\ldots,\varphi_k \in V^*$ 
such that $\langle g, \Delta(p)^k \rangle \ne 0$ and 
$\langle e_1^s, \Delta(p)\rangle \ne 0$. 
Applying Theorem~\ref{th:lift-ip-preserv}, we obtain with 
$q :=X_1^{n-s}\, p$ %= X_1^{n-s} \big( \varphi_1^{s} + \cdots + \varphi_{k}^{s} \big)$ 
that 
$$
           \langle h , q^{d} \rangle 
        = \langle f , M^*(q)^{d} \rangle 
        = \langle f , \Delta(p)^{d} \rangle 
        = \langle e_1^s, \Delta(p)\rangle^{d-k}\, \langle g, \Delta(p)^{k} \rangle\ne 0 .
$$
%where have used \eqref{eq:outdeg-triv} for the second equality. 
On the other hand,  by Theorem~\ref{eq:I-containlinforms},  
the padded polynomial 
$X_1^{n-s} \big( \varphi_1^{s} + \cdots + \varphi_{k}^{s} \big)$ 
is contained in~$\Det_n$, 
as $n\ge sk$ by assumption. 
Therefore, $h$ does not vanish on~$\Det_n$.
\end{proof}

%%%
\subsection{Building blocks and splitting technique}\label{sec:BBS} 

We construct as ``building blocks'' certain partitions that occur in $\IC[\Det_n]$. 
We achieve this by providing explicit tableau constructions and showing that 
the corresponding highest weight vectors do not vanish on certain tensors 
describing padded power sums. 
Then we apply Theorem~\ref{eq:I-containlinforms}. 
In particular, we prove this way that  
certain plethyms coefficients are nonzero. 
 
We first provide the proof of Proposition~\ref{pro:I-explicitHWV}, which 
deals with the case of even partitions and which was already stated in %the Overview 
Section~\ref{se:outline}. 

%\begin{proposition}\label{pro:I-explicitHWV}
%Let $n \ge k\ell$ and $k$ be even. 
%Then $(k\times \ell)^{\sharp nk}$ occurs in $\IC[\Det_n]_k$. 
%{\tt New notation ok?} 
%\end{proposition}

%\begin{proposition}%\label{pro:explicitHWV}
%Let $N$ be even and $n \ge Nd$.
%Then $(d \times N)^{\sharp nd}$ occurs in $\IC[\Det_n]_d$. 
%\end{proposition}

\begin{proof}[Proof of Proposition~\ref{pro:I-explicitHWV}]
Let $T$ denote the tableau of shape $k\times \ell$ with content $k\times \ell$
from Corollary~\ref{cor:ps-eval}. 
Suppose $n\ge k\ell$ and 
let $h := \kappa^k_{\ell,n}(v_T) \in \Sym^k\Sym^n V$ denote the lifting
of $v_T\in\Sym^k\Sym^\ell V$. 
Hence $h$ is a highest weight vector of weight $(k \times \ell)^{\sharp nk}$.  
Recall the corresponding linear automorphism  
$\Delta:=\Delta_{\ell,n}\colon\Sym^{\ell} V^* \to \Sym^{n} V^*$ 
from Lemma~\ref{le:newscaling}. 
Put
$p := X_1^\ell +\cdots+ X_k^\ell$ and note that 
$\Delta(p) = a X_1^\ell + b X_2^\ell + \cdots+ b X_k^\ell$ 
for some $a,b \ne 0$. 
Applying Theorem~\ref{th:lift-ip-preserv}, we obtain with 
$q :=X_1^{n-\ell}\, p$ that 
$\langle h , q^{k} \rangle= \langle v_T , M^*(q)^{k} \rangle  = \langle v_T , \Delta(p)^{k} \rangle$.
On the other hand, Corollary~\ref{cor:ps-eval} implies that  
$\langle v_T , \Delta(p)^{k} \rangle$ is nonzero. 
By Theorem~\ref{eq:I-containlinforms}, 
we have $X_1^{n}\in \Det_n$ since $n\ge k\ell$.
Therefore $h$, viewed as a degree~$d$ homogeneous polynomial function on $\Sym^n V^*$, 
does not vanish on $\Det_n$ and the assertion follows. 
\end{proof}

In order to handle partitions with odd parts, 
we use as further building blocks partitions obtained from 
rectangles by adding a single row and a single column. 

We postpone the proof of the following technical result 
to Section~\ref{sec:tableaux}.  
(It is based on an explicit construction of a highest weight vector.)

\begin{theorem}\label{cor:buildingblock}
Let $2 \leq b,c \leq m^2$ and let $n \geq 24  m^6$.
Then there exists an even $ i \leq 2m^4$, such  that 
\[
\la = b \times 1 + c \times i + 1\times j
\]
occurs in $\mathbb{C}[\Omega_n]_{3m^4}$
for $j = 3m^4n - b - ic$.
\end{theorem}

The splitting strategy in the following proof is a 
refinement of the one in~\cite{ik-panova:15}. 
The proof relies on Theorem~\ref{cor:buildingblock}
and on the semigroup property (Lemma~\ref{le:SGP}).

\begin{proposition}\label{thm:newoccurrence}
Given a partition $\la$ with $|\la|=nd$ 
such that there exists $m \geq 2$ with 
$\ell(\la)\leq m^2$, $m^{10} \leq |\bar\la|\leq md$, $n\geq 24m^6$, and $d>4m^6$. 
Then $\la$ occurs in $\IC[\Det_n]_d$.
\end{proposition}

\begin{proof}
Let $L:=\ell(\lambda)$ and $c_k$ denote the number of columns of length $k$ in~$\la$
for $1\le k \le L$. Let $K$ be the index $k \geq 2$, for which $c_k$ is maximal, 
i.e., $c_K = \max(c_k; k=2,\ldots, L)$. 
By assumption, we have $2 \leq K \leq m^2$ and 
$$
 m^{10} \le |\bar{\la}| = \sum_{k=2}^L (k-1) c_k \le c_K \sum_{k=2}^L (k-1) 
  \le c_K \frac{L^2}{2} \le c_K \frac{m^4}{2} ,
$$
hence $c_K \ge 2m^6$. 

The columns of odd length of $\la$ need a special treatment: 
let $S$ denote the set of integers $k\in\{2,\ldots,L\}$ for which $c_k$ is odd.
For $k \in S$ we define the partition 
\[
\omega_k := k \times 1 + K \times i_k , 
\]
where the even integer $i_k \leq 2m^4$ is taken from Theorem~\ref{cor:buildingblock}, 
so that $\omega_k^{\sharp 3nm^4}$ occurs in $\IC[\Det_n]_{3m^4}$.
(Here we have used the assumption $n\geq24m^6$.) 

Assume first that $K\not\in S$, that is, $c_K$ is even. 
Then we can split $\la$  vertically in rectangles as follows:
\begin{equation*}
\begin{split}
 \la &= 1 \times c_1 + \sum_{k=2 \atop k\not\in S \cup \{K\} }^{L} k \times c_k 
      + \sum_{k=2 \atop k\in S}^{L} k \times c_k + K \times c_K \\
   & = 1 \times c_1 + \sum_{k=2 \atop k\not\in S \cup \{K\} }^{L} k \times c_k 
      + \sum_{k=2 \atop k\in S}^{L} k \times (c_k-1)   + \sum_{k\in S} \omega_k 
        + K \times \Big(c_K - \sum_{k \in S} i_k \Big) .
\end{split}
\end{equation*}
If, for $k\leq L$, we set
$d_k := c_k$ if $k\not\in S \cup \{K\}$ and $d_k := c_k -1 $ if $k\in S$, and 
define 
$d_K:=c_K - \sum_{k \in S} i_k$, 
then the above can be briefly written as 
\begin{equation}\label{eq:SPLIT}
 \la = 1 \times c_1 + \sum_{k=2}^{L} k \times d_k + \sum_{k\in S} \omega_k .
\end{equation}
By construction, all $d_k$ are even. 
It is crucial to note that, using $i_k \le 2m^4$, 
$$
d_{K} = c_{K}-\sum_{k \in S}i_k \geq c_{K}-(m^2-1)\cdot 2m^4  \geq c_{K} - 2 m^6  \ge 0.
$$
The last inequality is due to our observation at the beginning of the proof. 

In the case where $K\in S$, we achieve the same decomposition as in \eqref{eq:SPLIT} 
with the modified definition
$d_K:=c_K - 1- \sum_{k \in S} i_k$. Here, as well $d_K \geq 0$ and all $d_k$ are even. 

We need to round down rational numbers to the next even number,
so for $a\in\IQ$ we define 
$\llfloor a \rrfloor := 2 \lfloor a/2 \rfloor$. 
Note that $\llfloor a \rrfloor \ge a - 2$ for all $a\in\IQ$. Hence 
$\llfloor n/k\rrfloor \geq n/k -2 \ge 2$ for all $2 \leq k \leq m^2$,
since $n \geq 4m^2$.  

Using division with remainder, let us write 
$d_k = q_k \llfloor \frac{n}{k}\rrfloor + r_k$ with $0\le r_k < \llfloor \frac{n}{k}\rrfloor$. 
Then we split 
$k\times d_k = q_k (k \times \llfloor \frac{n}{k} \rrfloor ) + k\times r_k$.
Since $d_k$ is even and $\llfloor n/k\rrfloor$ is even, $r_k$ is even as well. 
From \eqref{eq:SPLIT} we obtain that the partition 
\begin{equation}\label{eq:musplit}
\mu := 
\sum_{k=2}^{L} q_k ((k \times \llfloor n/k\rrfloor)^{\sharp nk}) 
+ \sum_{k=2}^{L} (k \times r_k)^{\sharp nk} 
 + \sum_{k\in S} \omega_k^{\sharp 3nm^4} 
\end{equation}
coincides with $\la$ in all but possibly the first row.

Since $\llfloor n/k\rrfloor \leq n/k$, $r_k \leq n/k$,
and both $\llfloor n/k\rrfloor$ and $r_k$ are even, 
Proposition~\ref{pro:I-explicitHWV} implies that 
$(k \times \llfloor n/k\rrfloor)^{\sharp nk}$ 
and $(k \times r_k)^{\sharp nk}$
occur as highest weights in $\IC[\Det_n]_k$.
Moreover, 
Theorem~\ref{cor:buildingblock} tells us that 
$\omega_k^{\sharp 3nm^4}$ 
occurs as a highest weight in $\IC[\Det_n]_{3m^4}$.
The semigroup property implies that $\mu$ occurs in $\IC[\Det_n]$.
\medskip 

\noindent {\bf Claim.} $|\mu| \le dn$. 
\medskip 

Let us finish the proof assuming the claim. 
If $|\mu| \le dn$, we can obtain $\la$ from $\mu$ by adding boxes to the first row of $\mu$. 
Note that $|\la|-|\mu|$ is a multiple of $n$. 
Since $(n)\in \IC[\Det_n]$, 
the semigroup property implies that $\la$ occurs in $\IC[\Det_n]_d$.
\smallskip 

It remains to verify the claim. 
From \eqref{eq:musplit} we get 
$$
 |\mu| \ \le\  \sum_{k=2}^L (q_k nk + nk + 3nm^4 ) .
$$
We have, using $\llfloor a \rrfloor \ge a -2$, 
$$
 q_k \le \frac{d_k}{\llfloor n/k \rrfloor} \le \frac{kd_k}{n-2k} .
$$
This implies 
$$
 |\mu| \ \le\ n \sum_{k=2}^L \Big( \frac{k^2 d_k}{n-2k} + k + 3m^4\Big) .
$$
Using $d_k \le c_k$ and $L\le m^2$, we get
$$
 |\mu| \ \le\  n \sum_{k=2}^L \frac{m^2}{n-2m^2} k c_k + n \sum_{k=2}^{m^2} k + 3nm^4 (m^2 -1) .
$$
Noting that $\sum_{k=2}^L k c_k = |\bar{\la}| + \la_2 \le 2 |\bar{\la}|$, we continue with 
\begin{eqnarray*}
 |\mu| &\le& \frac{nm^2}{n-2m^2} \cdot 2 |\bar{\la}| + 
             n \Big( \frac{m^2(m^2+1)}{2} + 3m^4(m^2 -1)\Big) \\
          &\le& \frac{nm^2}{12m^6-m^2} \cdot  |\bar{\la}| + 
             n \Big( 3m^6 -\frac52 m^4 + \frac12 m^2 \Big),
\end{eqnarray*}
where we have used $n>24m^6$ for the second inequality. 
Plugging in the assumptions $|\bar{\la}| \le dm$ and $d>4m^6$,  
we obtain
$$
 |\mu| \ \le\  \frac{dnm^3}{11m^6} + 3nm^6 \le 
  \frac{dn}{11} + 3nm^6 \le \frac{dn}{11} + \frac{3dn}{4} < dn ,
$$
which shows the claim and completes the proof.
\end{proof}

We can now finally complete the proof of our main result.

\begin{proof}[Proof of Theorem~\ref{th:main}] 
We may assume that $m\ge 2$, as the case $m=1$ is trivial. 
Suppose that $\la\vdash nd$ occurs in $\IC[Z_{n,m}]$ and 
$n \ge m^{25}$. 
Theorem~\ref{pro:I-KaLa} implies that $|\bar\la|\leq md$ and $\ell(\la)\leq m^2$. 

In the case of ``small degree'', where $n\ge md^2$, Proposition~\ref{pro:I-degreebound} 
implies that $\la$ occurs in $\IC[\Det_n]$. 

So we may assume that $d > \sqrt{n/m}$. 
In this case we have 
$d \ge \sqrt{m^{25}/m} = m^{12}$.
We conclude by two further case distinctions.

If $|\bar\la| < m^{10}$, 
we can apply Proposition~\ref{thm:verylongfirstrow} 
with $s:=m^{10}$ since 
$\la_2 \le |\bar{\la}| \le  s$,
$m^2 s^2 = m^{22} \le n$, 
and 
$m^2 s = m^{12}\le d$.
Thus $\la$ occurs in $\IC[\Det_n]_d$. 

Finally, 
if $|\bar\la| \geq m^{10}$, 
then the above Proposition~\ref{thm:newoccurrence} tells us that 
$\la$~occurs in $\IC[\Det_n]_d$.
\end{proof}

\section{Explicit constructions of tableaux and positivity of plethysms}%\label{sec:explicit}
\label{sec:tableaux}

The goal of this last section is to provide the proof of Theorem~\ref{cor:buildingblock}. 

In order to motivate the construction, we begin with a general reasoning.
Let $T$ be a tableau of shape~$\la$ with content $\out\times\inn$.  
For the sake of readability, we will use the natural numbers $1,\ldots,d$ as letters.
The set of boxes of $T$ is partitioned as 
$C_1 \cup\ldots \cup C_d$,
where $C_u$ denotes the set of boxes with the letter $u$. 
Note that $|C_u| = n$ for all~$u$. 
We denote by $C^1_u$ the subset of $C_u$ consisting of the boxes in singleton columns. 
On the other hand, the position set $[dn] = B_1 \cup\ldots \cup B_d$ is partitioned into the blocks
$B_u := \{ (u-1)n+ v \mid  1 \le v \le n \}$, where $|B_u| = n$ for all~$u$. 

We fix a map $s\colon [dn] \to [N]$, which defines the rank one tensor 
$\Phi=X_{s(1)}\otimes\ldots\otimes X_{s(dn)} \in \tensor^{dn} V^*$,
where $V:=\C^N$. %\IC^N$.  
Depending on~$s$, we denote by 
$B^1_u := B_u \cap s^{-1}(1)$ the set of positions in block $B_u$ 
that are mapped to~$1$ under the map~$s$. 
(Hence the positions in $B^1_u$ are the ones mapped to the basis vector~$X_1$.)

Recall from Theorem~\ref{eq:comb_contraction} that 
$\langle v_T,\Phi\rangle = (n!^d d!)^{-1} \sum_{\vartheta} \val_\vartheta (s)$, 
where the sum is over all bijections $\vartheta\colon\la\to [dn]$ respecting $T$.
The bijection $\vartheta$ respects the tableau~$T$ if 
there exists $\tau_\vartheta\in\aS_d$ such that 
$\vartheta(C_u) = B_{\tau_\vartheta(u)}$ for all $u$;
see\ Definition~\ref{def:respects}. 

If $\val_\vartheta(s) \ne 0$, then $s \circ \vartheta$ must map all boxes 
in singleton columns to $1$.  
This means
\begin{equation}\label{eq:key-ST}
 \forall u\quad \vartheta(C^1_u) \subseteq B^1_{\tau_\vartheta(u)} .
\end{equation}
For proving that $a_{\la}(\out[\inn]) >0$, we shall design 
$T$ and $s$ in such a way that 
$\val_\vartheta(s) \ge 0$ for all $\vartheta$ and 
there are only few $\vartheta$ with $\val_\vartheta(s) > 0$
(there must be at least one). 

Part of the strategy for realizing this can be described as follows. 

\begin{claim}\label{le:s=id}
Let $T$ be tableau of shape $\la$ and content $d\times n$,
where $C^1_u$ denotes the set of boxes with letter~$u$ in the singleton columns of~$T$.
Let $s\colon [dn] \to [N]$ and recall 
$B^1_u:= B_u \cap s^{-1}(1)$. 
Assume there is an integer~$D$ with $\ell(\la) \le D \le d$ such that 
$|C^1_1| = |B^1_1|,\ldots,|C^1_D| = |B^1_D| \le n-2$ 
are pairwise distinct numbers 
and $|C^1_u| > n-2$ for $D<u \le d$.  
Then for any $\vartheta\colon\la \to [dn]$ respecting~$T$ 
with $\val_\vartheta(s) \ne 0$, 
we have 
$\tau_\vartheta(u) = u$ for all $1\le u \le D$. 
\end{claim}

\begin{proof}
Assume $\val_\vartheta(s) \ne 0$ 
and write $\tau:= \tau_\vartheta$. Then \eqref{eq:key-ST} implies 
$|C^1_u| \le |B^1_{\tau(u)}|$ for all~$u$. 
For $u>D$ we have $|C^1_u| > n-2$, hence  $\tau(u) >D$, 
since $|B^1_{u'}| \le n-2$ for $u'\le D$. 
We conclude that $\tau$ permutes the set $[D]$. 
For $u\le D$, by assumption, 
the cardinalities $w(u) := |C^1_u| = |B^1_{u}|$ are pairwise distinct 
and \eqref{eq:key-ST} gives $w(u) \le  w(\tau(u))$.
This implies that $\tau(u)= u$ for $1\le u\le D$.
\end{proof}

By a concrete choice of a tableau $T$ and map $s$ we prove now the following. 

\begin{proposition}\label{prop:short_column}
Let $t \geq r$, $i \geq 2t+3$ be positive integers and let $\inn \geq i$ and $\out \geq 2t+i+1$. 
Let $\nu = (t+1) \times i + (r+1) \times 1 + (j)$, where $j=\out\inn - (t+1)i -(r+1)$. Then $a_{\nu}(\out[\inn]) >0$. 
\end{proposition}

\begin{proof}
%Due  to Lemma~\ref{le:KL-easy}(1) 
%Lemma~\ref{le:deglift}(1) 
We may assume that  $n= i$ and $d = 2t+i+1$.
%(see Lemma~\ref{le:dMmn} and \eqref{eq:deglift}).

Let $T$ be a tableau of shape $\nu$ labeled with the integers $1,2,3,\ldots,d$, each appearing $n$ times, as 
explained in Figure~\ref{fig1} for the case $t= 5$, $r=3$ and $i=13$. 
\begin{figure}[h]
\begin{center}
\begin{tikzpicture}
\node at (0,0) {
\ytableausetup{boxsize=3ex} %5ex}
$$\ytableaushort{{11}{12}{13}{14}{15}{16}{17}{18}{19}{20}{21}{22}{23}{24}111{...},11{10}{10}{10}{10}{10}{10}{10}{10}{10}{10}{10}{10},22299999999999,33338888888888,4444477777777,5555556666666}.$$
};
\node at (-1.25,1.8) {$\overbrace{\hspace{6.6cm}}^{i}$};
\draw [decorate,decoration={brace}] (2.65,1.025) -- (2.65,-0.5) node {};
\node at (3,0.2625) {\footnotesize $r$};
\draw [decorate,decoration={brace}] (-4.65,-1.517) -- (-4.65,1.025) node {};
\node at (-5,-0.246) {\footnotesize $t$};
\end{tikzpicture}
\end{center}
\caption{Prop.~\ref{prop:long_column}: $t= 5$, $r=3$, $i=13$, $d=24$, $n=13$, $D=10$, $dn=312$, $j=230$.}\label{fig1}
\end{figure}
Formally, 
if $1 \le k \le r$, the row $k+1$ of $T$ has $i+1$ boxes: $k+1$ boxes are labeled~$k$, 
and the remaining $i-k$ boxes are labeled $2t+1-k$. 
If $r < k \leq t$, then the row $k+1$ of $T$ has $i$ boxes: $k+1$ boxes are labeled $k$ and 
the remaining $i-k-1$ boxes labeled $2t+1-k$. 
The first row of $T$ starts with the first $i+1$ boxes labeled with 
$2t+1,\ldots,d=2t+i+1$, respectively, 
and all the remaining $j$ labels are put in the singleton columns of $T$ such that each integer 
in $1,\ldots,d$ appears exactly $n$ times. 
Note that each integer $1,\ldots,d$ appears in at least one singleton column, 
since $n \geq i \ge 2t+3$.

Put $D:=2t$. By construction, 
for any $1 \le u \leq D$ in $T$, 
$u$ appears in
row 1 and in a unique row~$k_u +1$ 
for some $1 \le k_u \leq t$. 
Let $\beta(u)$ denote the number of occurrences of $u$ in row $k_u + 1$. 
Note that $2 \le \beta(1) < \beta(2) < \ldots < \beta(D)$ by construction. 
Using the notation introduced before the proof, 
we have by construction
\begin{equation}\label{eq:IM}
 |C^1_u| = n - \beta(u) \le n-2 \mbox{ for $1\le u \le D$} 
  \quad \mbox{ and } \quad 
 |C^1_u| = n-1 \mbox{ for $D < u \le d$.}
\end{equation}

We consider now the tensor 
$$
 \Phi := \bigotimes_{u=1}^{D} \big(X_{k_u +1}^{\otimes \beta(u)} \otimes X_1^{\otimes (n-\beta(u))} \big)
     \otimes \bigotimes_{u=D+1}^{d} X_1^{n} ,
$$
which, more precisely, is defined by the map,  
$s\colon [dn] \to [N]$,  
$$
 s_{(u-1)n + v} = \left\{\begin{array}{cl} 
   k_u+1 & \mbox{ if $1 \le u \le D$ and $1 \le v \le \beta(u)$} \\
          1  & \mbox{ otherwise.}
   \end{array}\right.                            
$$
Using the notation introduced before the proof, we have 
$|C^1_u| = |B^1_u| = n - \beta(u)$ for $1\le u \le D$. 
Hence, by \eqref{eq:IM}, the assumptions of Claim~\ref{le:s=id} 
are satisfied. 
So if $\vartheta\colon\la\to [dn]$ respects~$T$ 
and gives a nonzero contribution to $\langle v_T,t\rangle$, then 
$\vartheta$~bijectively maps $C^1_u$ to $B^1_u$ 
for all $1 \le u \le D$ since $\tau_\vartheta(u)=u$, 
see~\eqref{eq:key-ST}.   
Hence a box $\square$ with the label $u$, which is not in the first row of~$T$, 
is mapped to a position in the $u$th block $B_u$. 
By the definition of $s$, this implies that 
$s(\vartheta(\square))$ equals 
the row number of~$\square$.  
This also holds true for boxes in singleton columns of $T$.
We conclude from Theorem~\ref{eq:comb_contraction}
that $\vartheta$~contributes the value~$1$ to $\langle v_T,\Phi\rangle$.

To complete the proof, it is sufficient to show the 
existence of some map $\vartheta\colon\la \to [dn]$ that 
respects~$T$, which is now obvious.  
In fact, there are 
$(n!)^{n-D} (n-D)! \prod_{u=1}^{D} \beta(u)! \cdot (\inn -\beta(u))!$
many maps $\vartheta$, that all contribute the value $1$ to $\langle v_T,\Phi\rangle$,
since $\vartheta$ can permute within every label the $X_{k_u+1}$ terms and the $X_1$ terms,
and $\vartheta$ can as well permute the labels $D+1,\ldots,n$.
\end{proof}

By generalizing this construction in the proof, 
we can show the following. 

\begin{proposition}\label{prop:long_column}
Let $t$, $r$ be positive integers, $i \in [ \frac{ (r+2t)^2 }{2t}, \frac{ (r+2t)^2 }{2t} + r+t+1 ]$, 
and let  $\inn > 6t+2r$ and $\out> r+2t +i $. Let 
$\nu = (t+1) \times i + (r+1) \times 1 + (j)$,
where $j=\out\inn - (r+1) - (t+1)i $. Then $a_{\nu}(\out[\inn]) >0$. 
\end{proposition}

\begin{proof}
If $r <t$ then we can directly apply Proposition~\ref{prop:short_column}, noticing that 
$$
 2t+2 <  \frac{ (1+2t)^2 }{2t} \le \frac{ (r+2t)^2 }{2t} \le  
 i \leq \frac{ (t+2t)^2 }{2t} +r+t+1 \leq \frac{11}{2} t +r +1 \leq 6t+r\leq n .
$$

Let now $r \geq t$.
The proof is similar to the proof of Proposition~\ref{prop:short_column}, 
so we describe a more general construction which applies in the case $r<t$ as well. 
Define $e := 2 ( \lfloor (r-1)/(2t) \rfloor +1 )$, so that $r \leq te \leq  r+2t-1 $ and $e$ is even. 
Put 
$$
 i' := (te+1) \frac{e}{2} \ \leq\  (r+2t) \frac{e}{2} \ \leq\  (r+2t) ( \lfloor \frac{r-1}{2t} \rfloor +1 )
 \ \leq\  (r+2t)\frac{(r+2t-1)}{2t} \ \leq\  i.
$$ 
We will prove the statement for $i=i'$. When $i>i'$, the tableau construction below can be modified 
by increasing the number of  appearances of the $t$ largest labels by $i-i' \leq r+t$ in the subtableau~$T'$ as defined below.
By assumption, 
$n>6t + 2r \ge te +2$ and 
$d > r +2t + i \ge te + i + 1$. 
Indeed, we will prove the statement for the more general case in which we do not require $n>6t + 2r$ and $d > r +2t + i$,
but only $n \ge te +2$ and $d \ge te + i + 1$.
%\CIcomment{I say this because the upcoming figure does not satisfy these other inequalities. Peter: Greta, do you agree? Greta: The Figure illustrates the proof, which is for a stronger but uglier bound. The statement is weaker for universality. So I agree with Christian's comment}
It suffices to prove the statement with $\inn = te+2$ and $\out =te+ i+1$.
%(see Lemma~\ref{le:dMmn} and \eqref{eq:deglift}).

Let $T$ be a tableau of shape $\nu$ filled with the labels $1,2,3,\ldots,d=te +i+1$, 
each number appearing $n=te+2$ times,
as in Figure~\ref{fig2} for the case $t=2,r=8,e=4, i= 18, n=10, d= 27$. 

\begin{figure}[h]
\begin{center}
\begin{tikzpicture}
\node at (0,0) {
\ytableausetup{boxsize=3ex} %4ex}
$$
\ytableaushort{{9}{10}{11}{12}{13}{14}{15}{16}{17}{18}{19}{20}{21}{22}{23}{24}{25}{26}{27}111{...},
8144445555588888888,
7223336666667777777,
6,
5,
4,
3,
2,
1} 
$$
};
\node at (-0.6,0.5) {$\underbrace{\hspace{9.2cm}}_{i}$};
\draw [decorate,decoration={brace}] (4,1.8) -- (4,0.75) node {};
\node at (4.3,1.275) {\footnotesize $t$};
\draw [decorate,decoration={brace}] (-5.8,-2.3) -- (-5.8,1.7672) node {};
\node at (-6.1,-0.2664) {\footnotesize $r$};
\end{tikzpicture}
\end{center}
\caption{Prop.~\ref{prop:long_column}: $t=2,r=8,e=4, i= 18, n=10, d= 27$.}\label{fig2}
\end{figure}

In the first row and in the first $i+1$ colums we have the labels $te+1,\ldots,te +i+1$.
In the first column and in the rows 2 to $r+1$ we have the labels $te,te-1\ldots,te-r+1$.
The remaining rectangular $t \times i$ subshape of $T$, denoted $T'$, 
consisting of the columns $2$ to $i+1$ and the rows $2$ to $t+1$,  
is filled with the remaining labels $1,\ldots,te$, so that each label appears a different number of times. 
More precisely, for each $1\le s \le te$, let the label $s$ appear in $T'$ exactly $s$ times and  
only in row 
$\min(\ell,2t-\ell+1)$, where $s \equiv \ell \pmod{2t}$, $1 \leq \ell \leq 2t$. 
(Note that the first row in $T'$, which we are referring to, is actually the second row in $T$.)
So the row $k$ of $T'$ contains the $e$ different labels $k, 2t+1-k, 2t+k, 4t+1-k, \ldots, t(e-2)+k, te +1-k$, 
each appearing that many times, adding up to the row length of 
$$
  \sum_{\alpha=1}^{e/2} \Big( (2(\alpha-1) t + k ) + (2\alpha t + 1 - k) \Big) 
  = (te+1)\frac{e}{2} = i . 
$$
The remaining labels of each kind are then put in the singleton boxes of $T$. 

As in Proposition~\ref{prop:short_column}, we show that the corresponding highest-weight vector $v_T$ 
in $\HWV_{\nu}(\Sym^\out \Sym^\inn V)$ is nonzero 
by contracting it with a particular monomial tensor~$\Phi$. 
For each label $u$, $1 \leq u \leq d$, let the associated monomial be
$$
  m_u = \bigotimes_{ \square \in T, \ \text{label}(\square)=u} X_{{\rm row}(\square)},
$$
where the product goes over all boxes of $T$ labeled $u$ and for each such box we take the variable $X$ 
whose index is the row the box is in. 
Again, let $\Phi:= \otimes_{u =1}^{\out} m_u$ be the tensor. 
We compute the contraction $\langle v_T, \Phi\rangle$ with Theorem~\ref{eq:comb_contraction}. 
%the combinatorial contraction of Section~\ref{subsec:combinatcontr}. 
The crucial observation is again that the labels of~$T$ below row 1 all appear a different number of times, 
so each $X_k$ appears in distinct nonzero degrees in the monomials $m_u$.
The rest of the proof of Proposition~\ref{prop:short_column} applies almost verbatim, 
and we see that $\langle v_T, \Phi \rangle \neq 0$:
Again, as in the case of $r \leq t$, we see that all nonzero summands in 
Proposition~\ref{eq:comb_contraction}
have the value 1 and there exists a nonzero summand that is trivial to construct.
\end{proof}

Finally we can complete the proof of the promised technical result. 

\begin{proof}[Proof of Theorem~\ref{cor:buildingblock}]
We apply Proposition~\ref{prop:long_column} with $r=b-1 \leq m^2-1$ and $t=c-1\leq m^2-1$. 
We have 
$$\frac{ (r+2t)^2}{2t} = \frac{ (b+2c-3)^2}{2(c-1)} \leq 
 \max\left( \frac{(b+1)^2}{2}, \frac{ (b+2m^2-3)^2}{2(m^2-1)}\right ) 
 \leq m^4,
$$
where we use the fact that $(b+2c-3)^2/2(c-1)$ is a convex function of $c$ and so attains its maximum at 
the end points of the interval $[2,m^2]$. 
We can then find an even integer~$i$ in the interval 
$[ \frac{ (r+2t)^2}{2t}, \frac{(r+2t)^2}{2t} + r+t+1] 
  \subseteq [ 1, m^4+2m^2]$. 
By Proposition~\ref{prop:long_column}, there exists
a highest weight vector~$f$ of weight $\nu = b\times 1 + c \times i + 1\times j'$ 
in $\Sym^{d}\Sym^{N}V$ for 
$$
 d := 3m^4 > 3m^2 + 2m^2 + m^4 \ge r +2t + i,\quad  
%and $
N:= 8m^2 > 6t+2r .
$$ 

We proceed now as for the proofs given in Section~\ref{sec:SmDeg}.
Recall the linear automorphism 
$\Delta_{N,n}\colon \Sym^N V^* \to \Sym^N V^*$  
from Lemma~\ref{le:newscaling}. 
By Proposition~\ref{cor:nonvanishing},  we have
$\langle f, \Delta_{N,n}(p)^d \rangle \ne 0$
for the power sum 
$p :=\varphi_1^{N} + \cdots + \varphi_{d}^{N}$
and generic $\varphi_i\in V^*$.
Moreover, by Theorem~\ref{eq:I-containlinforms}, 
$q:=X_1^{n-N} p$ is contained in $\Det_n$ 
for all $n \geq d N $, in particular for $n \geq 24 m^6$. 
Consider the lifting $h\in\Sym^d\Sym^n V$ of $f$;
it has the weight $\la =\nu^{\sharp dn}$ % + 1\times d(n-N)$ 
with $dn=3m^4 n$.
By Theorem~\ref{th:lift-ip-preserv} we have 
$\langle h,q^d \rangle = \langle f, M^*(q)^d \rangle 
 = \langle f, \Delta_{N,n}(p)^d \rangle \ne 0$.
%$H(q) = f(p) =F(p) \ne 0$. 
Therefore, $\la$~occurs in $\IC[\Det_n]_{3m^4}$.
\end{proof}

\bibliographystyle{plain}
\def\cprime{$'$}

\end{document}